\documentclass[journal]{IEEEtran}
\usepackage{epsfig,graphics,subfigure,psfrag,amsmath,cases}
\usepackage{latexsym,amssymb,amsmath,epsfig,subfigure,algorithm}
\usepackage{algorithmic,enumerate}
\usepackage{color}
\usepackage{url}
\usepackage{scrtime}
\usepackage{array}
\usepackage{datetime}
\author{\authorblockN{ Derrick Wing Kwan Ng,~\IEEEmembership{Member,~IEEE}, Mohammad Shaqfeh, ~\IEEEmembership{Member,~IEEE}  Robert Schober,~\IEEEmembership{Fellow,~IEEE}, and Hussein Alnuweiri, ~\IEEEmembership{Senior Member,~IEEE}\\
\thanks{Derrick Wing Kwan Ng is with the School of Electrical Engineering and Telecommunications, The University of New South Wales, Sydney, Australia.  Robert Schober is with both The University of British Columbia, Vancouver, Canada and the Institute for Digital Communications, Friedrich-Alexander-University Erlangen-N\"urnberg (FAU), Germany. Mohammad Shaqfeh and Hussein Alnuweiri are with  Texas A\&M University at Qatar,  Qatar.
This paper has been presented  in part at the IEEE WCNC 2014 \cite{CN:Kwan_layered_WCNC2014}.
This research was supported by the Qatar National Research Fund (QNRF), under project NPRP 5-401-2-161. Email: \{w.k.ng@unsw.edu.au,  rschober@ece.ubc.ca,  mohammad.shaqfeh@qatar.tamu.edu, hussein.alnuweiri@qatar.tamu.edu\}}}

}

\title{Robust Layered Transmission in Secure MISO Multiuser Unicast Cognitive Radio Systems}

\newtheorem{Thm}{Theorem}
\newtheorem{Lem}{Lemma}

\newtheorem{Prop}{Proposition}
\newtheorem{Remark}{Remark}
\DeclareMathOperator{\Tr}{Tr}
\DeclareMathOperator{\Rank}{\mathrm{Rank}}

\DeclareMathOperator{\zero}{\mathbf{0}}
\DeclareMathOperator{\vect}{\mathrm{vec}}

\DeclareMathOperator{\mino}{\mathrm{minimize}}
\newcommand{\textoverline}[1]{$\overline{\mbox{#1}}$}
\newcommand{\textovertide}[1]{$\widetilde{\mbox{#1}}$}
\DeclareMathOperator{\bigo}{\cal O}

\makeatletter

\newcommand{\Rmnum}[1]{\expandafter\@slowromancap\romannumeral #1@}
\makeatother
\newcommand{\abs}[1]{\lvert#1\rvert}
\newcommand{\norm}[1]{\lVert#1\rVert}

\newcolumntype{L}{>{\arraybackslash\raggedright}m{3.5cm}}
\newcolumntype{M}{>{\arraybackslash}m{7cm}}
\begin{document}

\maketitle
\begin{abstract}
This paper studies robust resource allocation algorithm design for a multiuser multiple-input single-output (MISO) cognitive radio (CR)  downlink communication network.  We focus on a  secondary system which provides wireless unicast
secure  layered video  information  to multiple single-antenna secondary receivers. The resource allocation algorithm
design is formulated as a non-convex optimization problem
for the minimization of the total transmit power at the secondary transmitter. The proposed framework
 takes into account a quality of service (QoS) requirement regarding video communication secrecy in the secondary system, the imperfection of the  channel state information (CSI) of potential eavesdroppers (primary receivers) at the secondary transmitter, and a limit for the maximum tolerable received interference power at the primary receivers. Thereby, the proposed problem formulation exploits the  \emph{self-protecting} architecture of layered
 transmission and artificial noise generation  to ensure communication secrecy. {The considered non-convex optimization problem is recast as a convex optimization
problem via semidefinite programming (SDP) relaxation. It is shown that the global optimal solution of
the original problem can be constructed by exploiting both the primal and the dual optimal solutions of the
SDP relaxed problem. Besides, two suboptimal resource allocation schemes  are proposed for the case
when the solution of the dual problem is unavailable for constructing the optimal solution.}
 Simulation results demonstrate  significant transmit power savings and robustness against CSI imperfection for the proposed optimal and suboptimal resource allocation algorithms employing layered transmission compared to  baseline schemes employing traditional single-layer transmission.
\end{abstract}
\section{Introduction}
\label{sect1}
 \IEEEPARstart{I}{n} recent years, the rapid  expansion of  high data  rate and secure multimedia services in wireless communication networks  has led to a tremendous demand for energy and bandwidth. The amount of video
traffic is expected to double annually in the near future and will be the main source of  wireless Internet traffic \cite{cisco_video}. As a result,  scalable video coding (SVC)  \cite{JR:Video_layers,JR:Video_layers2} has been proposed for video information encoding which provides high flexibility in resource allocation. In particular,  successive refinement coding (SRC) is one of the common multimedia SVC  techniques. In SRC, a video signal is encoded into a hierarchy of  multiple layers with unequal importance, namely one
 base layer and several enhancement layers. The base layer contains the essential information of the video with minimum video quality. The information embedded in each enhancement layer is used to successively  refine the description of the pervious layers.  The structure of layered transmission  facilitates the implementation of unequal error protection. In fact, SRC provides a high flexility to service providers since the transmitter can achieve a better resource utilization  by  allocating different
powers  to different information layers depending on the required video quality. Besides, layered transmission with SRC  has been implemented in some existing video standards such as  H.264/Moving Picture Experts Group (MPEG)-4 \cite{JR:video_overview}.

 Recently, resource allocation algorithm design for layered transmission has been pursued  for wireless communication systems. In  \cite{JR:Layer_transmission1}, power allocation for
layered transmission with successive enhancement was investigated. Subsequently, this study was extended to the joint design of rate and power allocation in \cite{JR:Layer_transmission2}.   A real-time adaptive resource allocation algorithm for layered video streaming was designed in  \cite{JR:Layer_transmission4} for multiple access networks. An information combing scheme for  double-layer video transmission over decode-and-forward wireless relay networks was proposed in \cite{JR:Layer_transmission3}.  {
Furthermore, in \cite{JR:layer_new_1}, a bandwidth allocation scheme was proposed to maximize the bandwidth utilization for scalable video services
over wireless cellular networks. The authors in \cite{JR:layer_new_2} investigated the resource allocation algorithm design for layer-encoded television signals for wideband communication systems.   Power allocation algorithms were proposed for amplify-and-forward and decode-and-forward communication networks with layered coding in \cite{CN:new_layer_2} and \cite{CN:new_layer_1}, respectively. In \cite{JR:layer_new_3}, a suboptimal multicast user grouping strategy was developed to exploit multiuser diversity in multiuser video transmission systems employing scalable video coding. However, the
resource allocation algorithms in  \cite{JR:Layer_transmission1}--\cite{JR:layer_new_3} were designed for single-antenna transmitters and/or for long-term average design objectives and may not be applicable to delay-sensitive applications and multiple-antenna systems.}

 In the past decades, multiple-input multiple-output (MIMO) technology has emerged as one of the most prominent solutions in reducing the system power consumption. In particular, MIMO provides  extra spatial degrees of freedom for resource allocation \cite{book:david_wirelss_com}--\nocite{JR:TWC_large_antennas,JR:MIMO_layered,JR:MIMO_layered2}\cite{CN:MISO_layer} which facilitates a trade-off between multiplexing and diversity.  On the other hand,  cognitive radio (CR) was proposed as a possible solution for improving spectrum utilization \cite{JR:CR_overview,JR:CR_1}.   CR enables a secondary system to dynamically access the spectrum of a primary  system  if  the
interference from the secondary system is controlled such that it does not  severely degrade the quality of service
(QoS) of the primary system  \cite{Report:CR}. However, the broadcast nature of CR networks makes them vulnerable to eavesdropping.  For instance, illegitimate users or misbehaving legitimate users  of a communication system may attempt to use high definition video  services without paying by overhearing the video signal. Conventionally,   secure communication employs  cryptographic encryption algorithms implemented in the application layer. However, the associated required secrete key distribution and management can be problematic or infeasible in wireless networks. Besides, with the development of  quantum computing, the commonly used encryption algorithms may become eventually breakable with a brute force approach.  As a result,  physical (PHY) layer
security \cite{Report:Wire_tap}--\nocite{JR:Artifical_Noise1,JR:Kwan_physical_layer}\cite{JR:Kwan_EE_PHY} has been proposed as a complement to the traditional secrecy methods for improving wireless transmission security. The merit of PHY layer security lies in the guaranteed  perfect secrecy of communication, even if the eavesdroppers have unbounded computational capability.
 In  \cite{Report:Wire_tap}, Wyner showed
 that a non-zero secrecy capacity,
defined as the maximum transmission rate at which an eavesdropper
is unable to extract any information from the received signal, can be achieved if
the desired receiver enjoys better channel conditions than the
eavesdropper.   In \cite{JR:Artifical_Noise1} and  \cite{JR:Kwan_physical_layer}, artificial noise generation
was exploited for multiple-antenna transmitters to weaken
 the information interception capabilities  of the eavesdroppers. In particular, artificial noise is transmitted concurrently with the information signal in  \cite{JR:Artifical_Noise1} and  \cite{JR:Kwan_physical_layer} for the maximization of the ergodic secrecy capacity and the outage secrecy capacity, respectively.
 In \cite{JR:Kwan_EE_PHY}, a joint power and subcarrier allocation algorithm was proposed for the maximization of the system energy efficiency of wideband communication systems  while providing communication secrecy.
{ The combination of CR and physical layer security was investigated  in  \cite{JR:CR_PHY_tuts}--\nocite{JR:CR_PHY_maga,JR:CR_PHY_1,JR:CR_PHY_2,JR:CR_PHY_3,JR:PHY_CR,JR:PHY_CR2,JR:PHY_CR3}\cite{JR:MOOP}.
In \cite{JR:CR_PHY_1}, the authors investigated the secrecy outage probability of CR systems in the presence of a passive eavesdropper.  In \cite{JR:CR_PHY_2} and \cite{JR:CR_PHY_3}, precoding and  beamforming schemes  were designed to ensure communication security for MIMO multiple eavesdropper (MIMOME) CR networks and cooperative  CR networks, respectively.  The authors in \cite{JR:PHY_CR} studied robust transmitter
designs for secure CR networks. In \cite{JR:PHY_CR2}, secure multiple-antenna transmission strategies were proposed to maximize the secrecy outage capacity of CR networks in slow fading. In \cite{JR:PHY_CR3}, the secrecy outage and diversity performances of CR systems were studied.
In \cite{JR:MOOP}, multiple objective optimization was adopted for CR networks to study the trade-off between  the interference leakage to the primary network and the transmit power of the secondary transmitter.
  However, the results in \cite{Report:Wire_tap}--\nocite{JR:PHY_CR2}\cite{JR:MOOP} did not exploit the properties of the targeted applications and may not be applicable for multimedia services.} Nevertheless, as soliciting multimedia services over the wireless medium becomes more popular, there is an emerging need for guaranteeing secure wireless video communication. In fact, as will be shown in this paper, the
layered information architecture  of video signals has a
\emph{self-protecting structure} which provides a certain robustness against  eavesdropping. To the best of the authors knowledge, exploiting the  layered transmission architecture of video signals for facilitating PHY layer security   has not been considered in the literature before. The notion  of  secure communication in layered (non-CR) transmission systems has been studied in our preliminary work in \cite{CN:Kwan_layered_WCNC2014}. Specifically, a power allocation algorithm was designed for the minimization of the transmit power under a communication secrecy constraint for a single video receiver. Yet, as availability of perfect CSI of the primary users at the secondary transmitter was assumed in \cite{CN:Kwan_layered_WCNC2014}, the resulting design advocates the generation of strong artificial noise/inteference to ensure secure video communication.  This may cause a significant performance degradation for the primary receivers if the results of \cite{CN:Kwan_layered_WCNC2014} are directly applied in CR networks having imperfect CSI  of the primary receivers.

{
 In this paper, we address the above issues and the corresponding contributions can be summarized as follows:
\begin{itemize}
\item We propose a non-convex  optimization problem formulation for the minimization of the total transmit power for layered video transmission to multiple secondary receivers. The proposed framework takes into account the  imperfection of the CSI of  the potential eavesdroppers (primary receivers) and exploits the   inherent \emph{self-protecting} structure of layered transmission for guaranteeing secure communication  to the secondary receivers and controlling the interference leakage to the multiple-antenna primary receivers.
    \item The considered non-convex optimization problem  is recast as a convex optimization problem via  semidefinite programming (SDP) relaxation. We prove that the global optimal solution of the original problem can be constructed based on the solutions of the  primal and the dual versions of the SDP relaxed problem.
    \item Two suboptimal resource allocation schemes are proposed for the case when the solution of the dual problem of the SDP relaxed problem is unavailable for construction of the optimal solution.
    \item Our simulation results show that the proposed algorithms exploiting layered transmission enable  significant transmit power savings in providing secure video communication for the secondary receivers compared to two baseline schemes employing traditional single-layer transmission.
\end{itemize}}

The rest of the paper is organized as follows. In Section \ref{sect:OFDMA_AF_network_model},
we outline the model for the considered secure layered video transmission. In Section \ref{sect:forumlation}, we formulate the
resource allocation algorithm design as an optimization
problem, and we solve this problem by  semidefinite programming relaxation in
Section \ref{sect:solution}. In Section \ref{sect:simulation}, we present numerical performance
results for the proposed optimal and suboptimal algorithms for secure video transmission. In Section \ref{sect:conclusion},
we conclude with a brief summary of our results.

\section{System Model}
\label{sect:OFDMA_AF_network_model}
In this section,  we  present the adopted system model for secure layered video transmission.

\subsection{Notation}
We use boldface capital and lower case letters to denote matrices and vectors, respectively. $\mathbf{A}^H$, $\Tr(\mathbf{A})$, $\mathbf{A}^{\frac{1}{2}}$, $\Rank(\mathbf{A})$, and $\det(\mathbf{A})$ represent the Hermitian transpose, trace, square-root, rank, and determinant of  matrix $\mathbf{A}$, respectively; $\vect(\mathbf{A})$ denotes the vectorization of matrix $\mathbf{A}$ by stacking its columns from left to right to form a column vector;
$\mathbf{A}\otimes \mathbf{B}$ denotes the Kronecker product of matrices $\mathbf{A}$ and $ \mathbf{B}$; $[\mathbf{B}]_{a:b,c:d}$ returns the $a$-th to the $b$-th rows and the $c$-th to the $d$-th columns block submatrix of $\mathbf{B}$;  $\mathbf{A}\succ \mathbf{0}$ and $\mathbf{A}\succeq \mathbf{0}$ indicate that $\mathbf{A}$ is a positive definite and a  positive semidefinite matrix, respectively; $\lambda_{\max}(\mathbf{A})$ denotes the maximum eigenvalue of matrix $\mathbf{A}$; $\mathbf{I}_N$ is the $N\times N$ identity matrix; $\mathbb{C}^{N\times M}$ denotes the set of all $N\times M$ matrices with complex entries; $\mathbb{H}^N$ denotes the set of all $N\times N$ Hermitian matrices.  The circularly symmetric complex Gaussian (CSCG) distribution is denoted by ${\cal CN}(\mathbf{m},\mathbf{\Sigma})$ with mean vector $\mathbf{m}$ and covariance matrix $\mathbf{\Sigma}$; $\sim$ indicates ``distributed as"; $\abs{\cdot}$, $\norm{\cdot}$, and $\norm{\cdot}_F$ denote the absolute value of a complex scalar, the Euclidean norm, and the Frobenius norm of a vector/matrix, respectively; $\mathrm{Re}\{\cdot\}$ denotes the real part of an input complex number and  $[x]^+=\max\{0,x\}$.


\subsection{Channel  Model}
We consider a CR secondary network. There are one secondary transmitter equipped with $N_{\mathrm{T}}>1$ antennas, a primary transmitter equipped with $N_{\mathrm{P}_\mathrm{T}}$ antennas,  $K$ legitimate secondary video receivers, and $J$ primary receivers. The secondary receivers and  the primary receivers  share the same spectrum concurrently, cf.  Figure \ref{fig:system_model}. The secondary receivers are low  complexity single-antenna devices for decoding the video signal. On the other hand, each primary receiver is
equipped with $N_{\mathrm{P}_\mathrm{R}}> 1$ antennas. We assume that $N_{\mathrm{T}}>N_{\mathrm{P}_\mathrm{R}}$ in this paper. In every time instant, the  transmitter conveys  $K$ video information signals to $K$ unicast secondary video  receivers.  The unicast scenario is applicable for on-demand video streaming service and provides high flexibility to the end-users. However, the transmitted video signals for each secondary receiver may be overheard by  primary receivers and unintended secondary receivers which share the same spectrum simultaneously.  \begin{figure}[t]
 \centering
\includegraphics[width=3.5in]{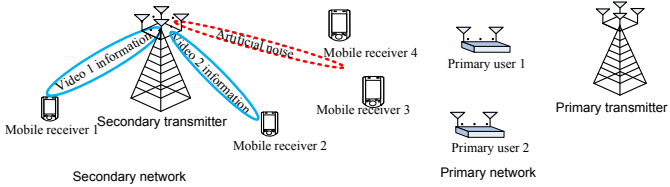}
 \caption{A CR network where $K=4$ secondary receivers  share the same spectrum with $J=2$ primary receivers. The secondary transmitter conveys video information to the  secondary receivers. The red dotted ellipsoid illustrates the  use of artificial noise for providing communication security in the secondary network.  } \label{fig:system_model}
\end{figure}
In
practice, it is possible that some receivers are malicious and
eavesdrop the video information of the other subscribers, e.g. a
paid multimedia video service, by overhearing the video signal
transmitted by the secondary transmitter.  As a result, the $J$ primary receivers and unintended secondary receivers are
potential eavesdroppers which should be taken into account for resource allocation algorithm design for providing secure
communication.   We focus on frequency flat fading channels. The downlink
received signals at secondary video receiver $k\in\{1,\ldots,K\}$ and primary receiver $j\in\{1,\dots,J\}$ are given by
{
\begin{eqnarray}
y_k&=&\mathbf{h}^{H}_k \mathbf{x}+ \mathbf{t}^{H}_k \sum_{j=1}^J\mathbf{d}_j+n_{\mathrm{s}_k}\\
\mbox{and}\,\,\mathbf{y}_{j}^{\mathrm{PU}}&=&\mathbf{P}^{H}_j  \sum_{j=1}^J\mathbf{d}_j+\mathbf{G}_j^H\mathbf{x}+\mathbf{z}_{\mathrm{s}_j},
\end{eqnarray}
respectively, where $\mathbf{x}\in\mathbb{C}^{N_{\mathrm{T}}\times1}$ denotes the transmitted signal vector of the secondary transmitter. The channel vector between the secondary transmitter and  secondary receiver $k$ is denoted by $\mathbf{h}_k\in\mathbb{C}^{N_{\mathrm{T}}\times1}$. The channel  matrix between the secondary transmitter and  primary receiver $j$ is denoted by $\mathbf{G}_j\in\mathbb{C}^{N_{\mathrm{T}}\times N_{\mathrm{P}_\mathrm{R}}}$.  $\mathbf{d}_j\in\mathbb{C}^{N_{\mathrm{P}_\mathrm{T}}\times1}$ denotes the precoded information data vector at the primary transmitter intended for primary receiver $j$. $\mathbf{t}_k\in\mathbb{C}^{N_{\mathrm{P}_\mathrm{T}}\times1}$  represents the channel vector between the primary transmitter and secondary receiver $k$ while $\mathbf{P}_j\in\mathbb{C}^{N_{\mathrm{P}_\mathrm{T}}\times N_{\mathrm{P}_\mathrm{R}}}$ is the channel matrix between the primary transmitter and primary receiver $j$.  The equivalent noise at the  secondary receivers, which comprises the joint effect of  the received interference from the primary transmitter, i.e., $\mathbf{t}^{H}_k \sum_{j=1}^J\mathbf{d}_j$,  and thermal noise, $n_{\mathrm{s}_k}$, is modeled as additive white Gaussian
noise (AWGN)  with zero mean and variance $\sigma_{\mathrm{s}_k}^2={\cal E} \{\abs{\mathbf{t}^{H}_k \sum_{j=1}^J\mathbf{d}_j}^2+\abs{n_{\mathrm{s}_k}}^2\}$.
 $\mathbf{z}_{\mathrm{s}_j}\sim{\cal CN}(\mathbf{0},\sigma_{\mathrm{PU}_j}^2\mathbf{I}_{N_{\mathrm{P}_\mathrm{R}}})$ is the AWGN at the primary receivers.}
 \begin{figure}[t]
 \centering
\includegraphics[width=3.5in]{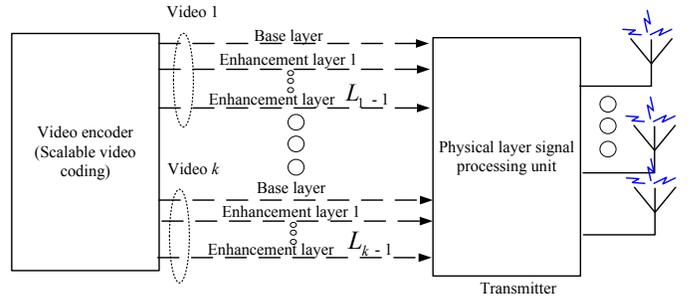}
 \caption{An illustration of layered coding for wireless video transmission.    } \label{fig:system_model_layer_coding}
\end{figure}
{
\begin{Remark}
In this paper, we assume that  the primary
network is a legacy system  and the primary transmitter does not actively participate in
transmit power control. Furthermore, we assume that the primary transmitter transmits a Gaussian signal and we focus on quasi-static fading channels such that all channel gains remain constant within the coherence time of the secondary system.  These assumptions justify modelling the interference from the primary transmitter to the secondary receivers as  AWGN. The total noise power, $\sigma_{\mathrm{s}_k}^2$, may be different  for different secondary receivers $k\in\{1,\ldots,K\}$. This model has been commonly adopted in the literature for resource allocation algorithm design \cite{JR:CR_beamforming_1}\nocite{JR:CR_phy}--\cite{CN:karama}.
\end{Remark}}
\subsection{Video Encoding and Artificial Noise}
Layered video encoding based on SRC is adopted to encode the video information, cf. Figure \ref{fig:system_model_layer_coding}. Specifically,
 the video source intended for secondary receiver $k$ is encoded  into $L_k$ layers at the secondary transmitter  and  the data rate of each layer is fixed, cf. H.264/SVC  \cite{JR:Video_layers,JR:Video_layers2}. The video information for  secondary receiver $k$ can be represented  as $\mathbf{S}_k=\big[s_{1,k},s_{2,k},\ldots,s_{l,k}\ldots,s_{L_k,k}\big]$, where $s_{l,k}\in\mathbb{C}$ denotes the video information of layer $l$ for secondary receiver $k$. For the video signal of receiver $k$,  the $L_k$ layers include one base layer, i.e., $s_{1,k}$, which can be decoded independently without utilizing the information
from the upper layers. Specifically, the base layer data includes
the most essential information of the video and can guarantee a basic video quality. The remaining $L_k-1$ layers, i.e., $\{s_{2,k},\ldots,s_{L_k,k}\}$, are  enhancement layers which are used to successively refine the  decoded lower layers. { In other words,
the video information embedded  in the enhancement layers cannot be retrieved independently; if
the decoding of the base layer fails, the information embedded in
the following enhancement layers is lost since it cannot be recovered.} Furthermore,  in order to provide communication security,  artificial noise is transmitted along with the information signals. Hence, the transmit symbol vector $\mathbf{x}$ can be expressed  as
\begin{eqnarray}
\mathbf{x}=\hspace*{-2mm}\underbrace{\sum_{k=1}^K \sum_{l=1}^{L_k}\mathbf{w}_{l,k} s_{l,k}}_{\mbox{layered video signals}}+\underbrace{\mathbf{v}}_{\mbox{artifical noise}},
\end{eqnarray}
where $\mathbf{w}_{l,k}\in\mathbb{C}^{N_\mathrm{T}\times 1}$ is the beamforming vector for the video information in layer $l$ dedicated to  desired receiver $k$. We note that  superposition coding is used to superimpose the $L_k$ video information layers. $\mathbf{v}\in\mathbb{C}^{N_\mathrm{T}\times 1}$ is an artificial noise vector generated to facilitate secure communication. In particular,  $\mathbf{v}$ is modeled as  a complex Gaussian random vector, i.e., $\mathbf{v}\sim {\cal CN}(\mathbf{0}, \mathbf{V})$, where $\mathbf{V}$ denotes the covariance matrix of the artificial noise. Hence,  $ \mathbf{V}$ is a  positive semidefinite Hermitian matrix, i.e.,  $\mathbf{V}\in \mathbb{H}^{N_\mathrm{T}}, \mathbf{V}\succeq \mathbf{0}$.

\section{Resource Allocation Algorithm Design}\label{sect:forumlation}
In this section, we present the adopted performance metrics and the problem formulation.

\subsection{Achievable Rate and Secrecy Rate}
\label{subsect:Instaneous_Mutual_information}
We assume that perfect CSI is available at the secondary video receivers. Besides,  successive interference cancellation (SIC) \cite{book:david_wirelss_com} is performed at the receivers for decoding video information. Thereby, before decoding the information in layer $l$,  the  receivers  first decode the video information in layers $\{1,\ldots,l-1\}$ and cancel the corresponding interference successively.  Therefore,  the instantaneous achievable rate between the transmitter and primary video receiver $k$ in  layer $l\in\{1,\ldots,L_k\}$ is given by
\begin{eqnarray}
C_{l,k}&=&\log_2\Big(1+\Gamma_{l,k}\Big),\,\,\\
\Gamma_{l,k}&=&\frac{\abs{\mathbf{h}^H_k\mathbf{w}_{l,k}}^2}{\Psi_{l,k}+\Tr(\mathbf{h}_k^H\mathbf{V}\mathbf{h}_k)+
\sigma_{\mathrm{s}_k}^2},\label{eqn:cap} \,\,\,\,
\mbox{and}\\
\Psi_{l,k}&=&\underbrace{ \overset{K}{{\underset{n\ne k}{\sum}}}\overset{L_n}{\underset{r=1}{\sum}} \abs{\mathbf{h}^H_k\mathbf{w}_{r,n}}^2}_{\mbox{multiuser interference}}+\underbrace{\overset{L_k}{\underset{t=l+1}{\sum}} \abs{\mathbf{h}^H_k\mathbf{w}_{t,k}}^2}_{\mbox{multilayer interference}},
\end{eqnarray}
where $\Gamma_{l,k}$ is the received signal-to-interference-plus-noise ratio (SINR) of layer $l$ at secondary video receiver $k$.

On the other hand, it is possible that secondary video receiver $t$ attempts to decode the video information intended for secondary receiver $k$ after decoding its own video information. Hence, secondary video receiver $t$ is treated as a potential eavesdropper with respect to the video information of secondary video receiver $k$.  The instantaneous achievable rate between the transmitter of secondary receiver $k$ and  secondary receiver $t$ in decoding layer $l\in\{1,\ldots,L_k\}$  is given by
\begin{eqnarray}\label{eqn:cap_unintended_users}
C_{l,k}^t&=&\log_2\Big(1+\Gamma_{l,k}^t\Big)\,\,\,\,
\mbox{and}\,\,\\
\Gamma_{l,k}^t&=&\frac{\abs{\mathbf{h}^H_t\mathbf{w}_{l,k}}^2}{I_{l,k}^t+\Tr(\mathbf{h}_t^H\mathbf{V}\mathbf{h}_t)+\sigma_{\mathrm{s}_t}^2}\label{eqn:cap_SINR} \\
I_{l,k}^t&=&\overset{K}{\underset{n\ne t}{\underset{n\ne k}{\sum}}}\overset{L_n}{\underset{r=1}{\sum}} \abs{\mathbf{h}^H_t\mathbf{w}_{r,n}}^2+\overset{L_k}{\underset{m=l+1}{\sum}} \abs{\mathbf{h}^H_t\mathbf{w}_{m,k}}^2 .
\end{eqnarray}
 It can be observed from (\ref{eqn:cap_unintended_users}) that layered transmission has a \emph{self-protecting structure}.  Specifically,  considering the first term in the denominator of (\ref{eqn:cap_SINR}), the higher layer information has the same effect as the artificial noise signal  $\mathbf{v}$ in protecting the information encoded in the
lower layers of the video signal. It is expected that by carefully optimizing the beamforming vectors of the higher information layers, a certain level of communication security can be achieved in the lower layers.

Besides, the transmitted video signals are also overheard by the primary receivers due to the broadcast nature of the wireless communication channel.
Therefore, the achievable rate between the transmitter and
primary receiver $j$ for decoding the $l$-th layer signal of  secondary receiver
$k$ can be represented as
\begin{eqnarray}\label{eqn:Capacity_eve}
\hspace*{-10mm}C_{{l,k}}^{\mathrm{PU}_j}\hspace*{-2.5mm}&=&
\hspace*{-2.5mm}\log_2\det\hspace*{-0.5mm}\Big(\mathbf{I}_{N_{\mathrm{P}_\mathrm{R}}}\hspace*{-0.5mm}+\hspace*{-0.5mm}
\mathbf{\Lambda}_{j,k}^{-1}
\mathbf{G}_j^H\mathbf{w}_{l,k}\mathbf{w}_{l,k}^H\mathbf{G}_j\hspace*{-0.5mm}\Big)\, \mbox{where}\\ \label{eqn:Sigma}
\mathbf{\Lambda}_{j,k}\hspace*{-2.5mm}&=&\hspace*{-2.5mm}\mathbf{\Sigma}_{j}+\underbrace{\sum_{n\ne k}^K \sum_{r=1}^{L_n}\mathbf{G}_j^H\mathbf{w}_{r,n}\mathbf{w}_{r,n}^H\mathbf{G}_j}_{\mbox{multiuser interference}}\notag\\
\hspace*{-2.5mm} &+& \hspace*{-2.5mm}\underbrace{\sum_{m=l+1}^{L_k}\mathbf{G}_j^H\mathbf{w}_{m,k}\mathbf{w}_{m,k}^H\mathbf{G}_j}_{\mbox{multilayer interference}},\\
\mathbf{\Sigma}_{j}\hspace*{-2.5mm}&=&\hspace*{-2.5mm}\mathbf{G}_j^H\mathbf{V}\mathbf{G}_j+\sigma_{\mathrm{PU}_j}^2\mathbf{I}_{N_{\mathrm{P}_\mathrm{R}}}\succ \zero.
\end{eqnarray}
In practice, the behavior of the primary receivers  cannot be fully controlled by the secondary transmitter and it is possible that some primary receivers are malicious and eavesdrop the video information intended for the secondary receivers. { Hence, for ensuring communication security, the primary receivers are also treated as potential eavesdroppers
who attempt to decode the messages intended for all $K$
desired secondary receivers. Thereby, the primary user may be able to first decode the channel coded and modulated data symbols of the enhancement layers, i.e., $s_{l,k},l\in\{2,\ldots,L_k\}$, and then remove the associated interference before decoding the base layer, i.e., $s_{1,k}$, and retrieving the embedded video information\footnote{ We note that without knowledge of the base layer, an eavesdropper cannot reconstruct the video signal based on $s_{l,k},l\in\{2,\ldots,L_k\},$ because of the layered video coding. However, knowledge of $s_{l,k},l\in\{2,\ldots,L_k\}$, is beneficial for channel decoding of  $s_{1,k}$.  }.   Therefore, we focus on the worst-case scenario regarding the
decoding capability of the primary receivers for providing communication security to the secondary receivers. In particular, we assume that primary receiver $j$ performs
SIC to remove all multiuser interference and the multilayer interference from the upper layers before decoding the message of layer $l$ of secondary receiver $k$.}

As a result, the achievable rate in (\ref{eqn:Capacity_eve}) for decoding the first layer is bounded above by
\begin{eqnarray}\label{eqn:Capacity_eve_up}
\widetilde{C}_{{1,k}}^{\mathrm{PU}_j}=\log_2\det\Big(\mathbf{I}_{N_{\mathrm{P}_\mathrm{R}}}+\mathbf{\Sigma}_{j}^{-1}
\mathbf{G}_j^H\mathbf{w}_{1,k}\mathbf{w}_{1,k}^H\mathbf{G}_j\Big).
\end{eqnarray}
Thus, the secrecy rate  \cite{JR:Artifical_Noise1} between the transmitter
and secondary receiver $k$ on layer $1$  is given by
\begin{eqnarray}\label{eqn:secrecy_cap}
C_{\mathrm{sec}_{1,k}}=\Big[C_{1,k} -\underset{\overset{\forall t\ne k}{\forall j}}{\max}\,\{C_{1,k}^t,\widetilde{C}_{{1,k}}^{\mathrm{PU}_j}\}\Big]^+.
\end{eqnarray}
$C_{\mathrm{sec}_{1,k}}$  quantifies
the maximum achievable data rate at which a transmitter can reliably send  secret
information on layer $1$ to secondary receiver $k$ such that the potential eavesdroppers are unable to decode the received signal \cite{Report:Wire_tap}.


\subsection{Channel State Information}
 In this paper, we focus on a Time Division
Duplex (TDD)  communication system with slowly time-varying channels.  At the
beginning of each time slot,  handshaking is performed between the secondary transmitter and the secondary receivers. As a result, the downlink CSI of the secondary transmitter to the secondary receivers can be obtained by measuring the uplink training sequences embedded in the handshaking signals. Thus, we assume that the secondary-transmitter-to-secondary-receiver fading gains, $\mathbf{h}_k$, can be reliably estimated at the secondary transmitter with negligible estimation error.

On the other hand,   the primary receivers may not directly interact with the secondary transmitter. Besides, the primary receivers may be silent for a long period of time due to bursty data transmission. As a result,  the CSI of the primary receivers  can  be  obtained only occasionally at the secondary transmitter when the primary receivers communicate with the primary transmitter. Hence, the CSI for the idle  primary receivers may be outdated when the secondary transmitter performs resource allocation.  We adopt a deterministic model \cite{JR:Robust_error_models1}--\nocite{JR:Robust_error_models2,JR:CSI-determinisitic-model,JR:Kwan_secure_imperfect}\cite{JR:SWIPT_DAS}  to characterize the impact of the CSI imperfection on  resource allocation design. The  CSI of the link between the secondary transmitter
and primary receiver $j$ is modeled as
\begin{eqnarray}\label{eqn:outdated_CSI}
\mathbf{G}_j&=&\mathbf{\widehat G}_j + \Delta\mathbf{G}_j,\,   \forall j\in\{1,\ldots,J\}, \mbox{   and}\\
{\Psi }_j&\triangleq& \Big\{\Delta\mathbf{G}_j\in \mathbb{C}^{N_{\mathrm{P}_\mathrm{R}}\times N_{\mathrm{T}}}  :\norm{\Delta\mathbf{G}_j}_F^2 \le \varepsilon_j^2\Big\},\forall j,\label{eqn:outdated_CSI-set}
\end{eqnarray}
where $\mathbf{\widehat G}_j\in\mathbb{C}^{N_{\mathrm{P}_\mathrm{R}}\times N_{\mathrm{T}}}$ is the  matrix CSI estimate of the channel of primary receiver $j$ that is available at the secondary transmitter. $\Delta\mathbf{G}_j$ represents the unknown channel uncertainty due to the time varying nature of the channel during transmission. In particular,  the continuous set ${\Psi }_j$ in (\ref{eqn:outdated_CSI-set})  defines a continuous space spanned by all possible channel uncertainties and $\varepsilon_j$ represents the maximum value of the norm of the CSI estimation error matrix  $ \Delta\mathbf{G}_j$ for primary receiver $j$.
\subsection{Optimization Problem Formulation}
\label{sect:cross-Layer_formulation}
The system design objective is to minimize
the total transmit power of the secondary transmitter while providing QoS for both the secondary receivers and the primary receivers\footnote{{
We note that the performance of the considered system serves as an upper bound for the performance of a system where also the CSI of the secondary network is imperfect. The study of the impact of imperfect CSI of the secondary network on  performance is left for  future work.}}. The optimal resource allocation policy $\{\mathbf{w}_{l,k}^*,\mathbf{V}^*\}$ can be obtained by solving
\begin{eqnarray}
\label{eqn:cross-layer}&&\hspace*{5mm} \underset{\mathbf{V}\in \mathbb{H}^{N_\mathrm{T}},\mathbf{w}_{l,k}
}{\mino} \,\, \sum_{k=1}^K\sum_{l=1}^{L_k}\norm{\mathbf{w}_{l,k}}^2+\Tr(\mathbf{V})\\
\notag \mbox{s.t. } \hspace*{-1mm}&&\hspace*{-5mm}\mbox{C1: }\notag \Gamma_{l,k} \ge \Gamma_{\mathrm{req}_{l,k}}, \forall l,\forall k, \\
\hspace*{-1mm}&&\hspace*{-5mm}\mbox{C2: } \notag\Gamma_{1,k}^t  \le \Gamma_{\mathrm{tol}}, \forall t\ne k, t\in\{1,\ldots,K\} , \\
\hspace*{-1mm}&&\hspace*{-5mm}\mbox{C3:}\,\, \max_{\norm{\Delta\mathbf{G}_j}_F\in {\Psi}_j}\,\, \Tr\Big(\mathbf{G}_j^H(\mathbf{V}+\sum_{k=1}^K \sum_{l=1}^{L_k} \mathbf{w}_{l,k}\mathbf{w}_{l,k}^H)\mathbf{G}_j\Big)\notag\\
\hspace*{-1mm}&&\hspace*{-5mm}\le P_{\mathrm{I}_j},\forall j\in\{1,\ldots,J\},\notag\\
\hspace*{-1mm}&&\hspace*{-5mm}\notag\mbox{C4:}\,\, \max_{\norm{\Delta\mathbf{G}_j}_F\in {\Psi}_j}\,\, \widetilde{C}_{{1,k}}^{\mathrm{PU}_j}\le R_{\mathrm{Eav}_{j,k}}\notag,\forall j, \forall k,\\
\hspace*{-1mm}&&\hspace*{-5mm}\mbox{C5:}\,\, \mathbf{V}\succeq \mathbf{0}\notag.
\end{eqnarray}
 Here, $\Gamma_{\mathrm{req}_{l,k}}$ in C1 is the minimum required SINR  for decoding layer $l$ at receiver $k$.
 In C2, $\Gamma_{\mathrm{tol}}$ denotes the  maximum tolerated received SINR of layer $1$ at the unintended secondary receivers for decoding layer $1$ of a video signal intended for another receiver.  { Since layered video coding is employed,    it is sufficient to protect the  first layer of each video signal of each secondary receiver against  eavesdropping. In other words, the video information embedded in the enhancement layers is secure as long as the video information encoded  in the base layer is secure.}
 \begin{figure*}[!t]
\setcounter{equation}{17}
\begin{eqnarray}
\label{eqn:SDP}&&\hspace*{6mm} \underset{\mathbf{W}_{l,k},\mathbf{V}\in \mathbb{H}^{N_\mathrm{T}}
}{\mino} \,\, \sum_{k=1}^K\sum_{l=1}^{L_k}\Tr(\mathbf{W}_{l,k})+\Tr(\mathbf{V})\\
\notag \mbox{s.t. } \hspace*{-1mm}&&\hspace*{-5mm}\mbox{C1: }\notag \frac{\Tr(\mathbf{H}_k\mathbf{W}_{l,k})}{ \Tr\Big(\mathbf{H}_k \big(\overset{K}{\underset{n\ne k}{\sum}}\overset{L_n}{\underset{r=1}{\sum}}\mathbf{W}_{r,n}+\overset{L_k}{\underset{m=l+1}{\sum}} \mathbf{W}_{m,k}\big)\Big)+\Tr(\mathbf{H}_k\mathbf{V})+\sigma_{\mathrm{s}_k}^2} \ge \Gamma_{\mathrm{req}_{l,k}}, \forall l,\forall k, \\
\hspace*{-1mm}&&\hspace*{-5mm}\mbox{C2: } \notag\frac{\Tr(\mathbf{H}_t\mathbf{W}_{1,k})}{ \Tr\Big(\mathbf{H}_t\big(\overset{K}{\underset{n\ne t}{\underset{n\ne k}{\sum}}}\overset{L_n}{\underset{r=1}{\sum}}\mathbf{W}_{r,n}+\overset{L_k}{\underset{m=2}{\sum}}\mathbf{W}_{m,k}\big)\Big) +\Tr(\mathbf{H}_t\mathbf{V})+\sigma_{\mathrm{s}_k}^2} \le \Gamma_{\mathrm{tol}}, \forall t\ne k, t\in\{1,\ldots,K\} , \\
\hspace*{-1mm}&&\hspace*{-5mm}\mbox{C3:}\,\,\max_{\norm{\Delta\mathbf{G}_j}_F\in {\Psi}_j}\,\, \Tr\Big(\mathbf{G}_j^H\big(\mathbf{V}+\sum_{k=1}^K \sum_{l=1}^{L_k} \mathbf{W}_{l,k}\big)\mathbf{G}_j\Big)\le P_{\mathrm{I}_j},\forall j\in\{1,\ldots,J\},\notag\\
\hspace*{-1mm}&&\hspace*{-5mm}\notag\mbox{C4:}\,\, \max_{\norm{\Delta\mathbf{G}_j}_F\in {\Psi}_j}\,\, \log_2\det\Big(\mathbf{I}_{N_{\mathrm{P}_\mathrm{R}}}+\mathbf{\Sigma}_j^{-1}\mathbf{G}_j^H
\mathbf{W}_{1,k} \mathbf{G}_j\Big)\le  R_{\mathrm{Eav}_{j,k}}\notag,\forall j,\forall k,\\
\hspace*{-1mm}&&\hspace*{-5mm}\mbox{C5:}\,\, \mathbf{V}\succeq \mathbf{0},\notag\quad\quad\mbox{C6:}\,\, \mathbf{W}_{l,k}\succeq \mathbf{0},\forall k,l,\notag\quad\quad\mbox{C7:}\,\, \Rank(\mathbf{W}_{l,k})\le 1,\forall k,l,
\end{eqnarray}\hrulefill
\end{figure*}  C3 is the interference temperature constraint \cite{JR:MIMO_Robust_CR}. Specifically, the secondary transmitter is required to control the transmit power such that the maximum received interference power at primary receiver $j$ is less than a given interference temperature  $P_{\mathrm{I}_j}$,  despite the imperfection of the CSI. On the other hand,  although constraint C3 restricts   the total received power at the primary receivers, it does not necessarily guarantee communication security against eavesdropping by the primary receivers, especially when $P_{\mathrm{I}_j}$ is not zero.  Thus, we focus on the worst-case scenario for  robust secure communication  design by imposing constraint\footnote{In general, constraint C3 is not a subset of constraint C4 or vice versa and thus has to be treated explicitly.  } C4. Since any secondary receiver could  be chosen as an eavesdropping target of primary receiver $j$ and layered transmission is adopted,
the upper limit $R_{\mathrm{Eav}_{j,k}}$ is imposed in C4 to restrict the achievable rate of primary receiver $j$,  if it  attempts to decode the video base layer of secondary receiver $k,\forall k$.
 { In this paper, we do not  maximize the secrecy rate of video delivery as this does not necessarily lead to a power efficient resource allocation. Yet, the problem formulation in (\ref{eqn:cross-layer})
guarantees  a minimum secrecy rate for layer $1$, i.e., the base layer, of the video signal intended for secondary receiver $k$, i.e., $C_{\mathrm{sec}_{1,k}}\ge \Big[C_{1,k} -\underset{\overset{\forall t\ne k}{\forall j}}{\max}\,\{\log_2(1+\Gamma_{\mathrm{tol}}),R_{\mathrm{Eav}_{j,k}}\}\Big]^+$.
Besides, the video information of layer $2$ to layer $L_k$ is secured when layer $1$ cannot be decoded by the potential eavesdroppers. } Finally, C5
 and $\mathbf{V}\in \mathbb{H}^{N_\mathrm{T}}$  are imposed such that $\mathbf{V}$ satisfies the requirements for a covariance matrix.

{
  \begin{Remark}
 We would like to emphasize that the layered transmission approach has two major
advantages compared to single-layer transmission. First, the video quality increases with the
number of decoded layers.  In practice, the intended video receivers may belong to different classes with different numbers of video layers and different QoS requirements.
 For instance, the secondary video receivers may belong to one of two categories,  namely
\emph{premium video receivers} and \emph{regular video receivers}, based on the subscribed services. Specifically,  the secondary transmitter may be required  to  guarantee the signal quality  of all video layers for premium secondary receivers (i.e., the secrecy rate of the first layer and the data rate of the enhancement layers) while it may guarantee only the basic signal quality  (i.e., the secrecy rate of the first layer) of videos for regular receivers.  Thereby, the desired premium users may
be charged a higher subscription fee for higher video quality. Second, the  \emph{self-protecting} architecture of layered
 transmission enables a more power efficient resource allocation under physical layer security constraints. In
particular, instead of protecting the entire encoded video signal as in single-layer transmission,
in layered transmission, the transmitter has to protect only the most important part of the video,
i.e., the base layer, to provide communication security.
  \end{Remark}}

{
 \begin{Remark}
In this paper, we assume that problem (\ref{eqn:cross-layer}) is feasible for resource allocation algorithm design. In practice, the feasibility of the problem depends on the channel condition and the QoS requirements of the receivers.   If the problem is infeasible,  user scheduling can be performed at a higher layer to temporarily exclude some users from being served so as to improve the problem feasibility. However, scheduling design is beyond the scope of this paper. Interested readers may refer to \cite{JR:Kwan_AF}\nocite{JR:scheduling_policy}--\cite{JR:scheduling_policy2} for a detailed discussion of scheduling algorithms.
 \end{Remark}}

\section{Solution of the Optimization Problem} \label{sect:solution}
The optimization problem in (\ref{eqn:cross-layer}) is a non-convex quadratically constrained quadratic program (QCQP). In particular, the non-convexity of the considered problem is due to
constraints C1, C2, and C4. Besides, constraints C3 and C4 involve  infinitely many inequality constraints  due to the continuity of the CSI uncertainty sets, ${\Psi }_j,j\in\{1,\ldots,J\}$.
In order to derive an efficient
 resource allocation algorithm for the considered problem, we first rewrite the original problem to avoid the non-convexity associated with constraints C1 and C2. Then, we convert the infinitely many constraints in C3 and C4 into an equivalent finite number of  constraints. Finally, we use semi-definite programming relaxation (SDR) to obtain the resource allocation solution for the reformulated problem.

\subsection{Problem Transformation} \label{sect:solution_dual_decomposition}
First, we rewrite problem (\ref{eqn:cross-layer})  in an equivalent form as in \eqref{eqn:SDP}, where  $\mathbf{H}_k=\mathbf{h}_k\mathbf{h}^H_k$ and  $\mathbf{W}_{l,k}=\mathbf{w}_{l,k}\mathbf{w}^H_{l,k}$. We note that $\mathbf{W}_{l,k}\succeq \mathbf{0}$, $\mathbf{W}_{l,k}\in \mathbb{H}^{N_\mathrm{T}}$, and $\Rank(\mathbf{W}_{l,k})\le 1,\forall l,k,$ in (\ref{eqn:SDP}) are imposed to guarantee that $\mathbf{W}_{l,k}=\mathbf{w}_{l,k}\mathbf{w}^H_{l,k}$ holds after optimization.
{
Next, to handle the infinitely many constraints in C3, we introduce a Lemma which will allow us to convert them into a finite number of linear matrix inequalities
(LMIs).
\begin{Lem}[S-Procedure \cite{book:convex}] Let a function $f_m(\mathbf{x}),m\in\{1,2\},\mathbf{x}\in \mathbb{C}^{N\times 1},$ be defined as
\begin{eqnarray}
f_m(\mathbf{x})=\mathbf{x}^H\mathbf{A}_m\mathbf{x}+2 \mathrm{Re} \{\mathbf{b}_m^H\mathbf{x}\}+c_m,
\end{eqnarray}
where $\mathbf{A}_m\in\mathbb{H}^N$, $\mathbf{b}_m\in\mathbb{C}^{N\times 1}$, and $c_m\in\mathbb{R}$. Then, the implication $f_1(\mathbf{x})\le 0\Rightarrow f_2(\mathbf{x})\le 0$  holds if and only if there exists an $\omega\ge 0$ such that
\begin{eqnarray}\omega
\begin{bmatrix}
       \mathbf{A}_1 & \mathbf{b}_1          \\
       \mathbf{b}_1^H & c_1           \\
           \end{bmatrix} -\begin{bmatrix}
       \mathbf{A}_2 & \mathbf{b}_2          \\
       \mathbf{b}_2^H & c_2           \\
           \end{bmatrix}          \succeq \mathbf{0},
\end{eqnarray}
provided that there exists a point $\mathbf{\hat{x}}$ such that $f_k(\mathbf{\hat{x}})<0$.
\end{Lem}
Now, we apply Lemma 1 to constraint C3. In particular, we define $\mathbf{\widehat g}_j=\vect(\mathbf{\widehat G}_j)$, $\mathbf{\Delta g}_j=\vect(\Delta\mathbf{G}_j)$,  $\overline{\mathbf{W}}_{l,k}=\mathbf{I}_{N_{\mathrm{P}_\mathrm{R}}}\otimes\mathbf{W}_{l,k}$, and $\overline{\mathbf{V}}=\mathbf{I}_{N_{\mathrm{P}_\mathrm{R}}}\otimes\mathbf{V}$.  By exploiting  the fact that $\norm{\mathbf{\widehat G}_j}_F^2\le\varepsilon_j^2 \Leftrightarrow \Delta\mathbf{g}_j^H \Delta\mathbf{g}_j\hspace*{-1mm}\le\hspace*{-1mm} \varepsilon_j^2$, then we have
\begin{eqnarray}
\hspace*{-4.8mm}&&\hspace*{-0.8mm}\norm{\mathbf{\widehat G}_j}_F^2\le\varepsilon_j^2 \\
\Rightarrow\,\hspace*{-4.8mm}&&\hspace*{-0.8mm} \mbox{{C3}: }\hspace*{-1mm}0 \ge\hspace*{-1mm}  \max_{\Delta\mathbf{g}_j\in {\Psi}_j} \Delta\mathbf{g}_j^H\Big(\sum_{k=1}^K\sum_{l=1}^{L_k}\overline{\mathbf{W}}_{l,k}+\overline{\mathbf{V}}\Big)\Delta\mathbf{g}_j\notag\\
&&\hspace*{-2mm}+
2\mathrm{Re}\Big\{\mathbf{\widehat g}_j^H\Big(\sum_{k=1}^K\sum_{l=1}^{L_k}\overline{\mathbf{W}}_{l,k}+\overline{\mathbf{V}}\Big)\Delta\mathbf{g}_j\Big\}
\notag\\
&&\hspace*{-2mm}+\mathbf{\widehat g}_j^H\Big(\sum_{k=1}^K\sum_{l=1}^{L_k}\overline{\mathbf{W}}_{l,k}+\overline{\mathbf{V}}\Big)\mathbf{\widehat g}_j- P_{\mathrm{I}_j},\forall j\in\{1,\ldots,J\},\notag
\end{eqnarray}
if and only if there exists an $\omega_j\ge 0$ such that the following
LMIs constraint holds:
\begin{eqnarray}
\hspace*{-15mm}&&\label{eqn:LMI_C3}\hspace*{-5mm}\mbox{\textoverline{C3}: }\mathbf{S}_{\mathrm{\overline{C3}}_j}({\mathbf{W}}_{l,k},{\mathbf{V}}, \omega_j)\notag\\
&=&
         \begin{bmatrix}
       \omega_j\mathbf{I}_{N_{\mathrm{P}_\mathrm{R}}N_{\mathrm{T}}}-\overline{\mathbf{V}}& -\overline{\mathbf{V}}\mathbf{\widehat g}_j          \\
       -\mathbf{\widehat g}_j^H \overline{\mathbf{V}}
        & -\omega_j\varepsilon_j^2 +P_{\mathrm{I}_j} - \mathbf{\widehat g}_j^H \overline{\mathbf{V}} \mathbf{\widehat g}_j        \\
           \end{bmatrix}\notag\\
           &-& \mathbf{U}_{\mathbf{g}_j}^H\Big(\sum_{k=1}^K\sum_{l=1}^{L_k}\overline{\mathbf{W}}_{l,k}\Big)\mathbf{U}_{\mathbf{g}_j}\succeq \mathbf{0}, \forall j,
          \end{eqnarray}
where $\mathbf{U}_{\mathbf{g}_j}=\big[\mathbf{I}_{N_{\mathrm{P}_\mathrm{R}}N_{\mathrm{T}}},\,\, \mathbf{\widehat g}_j\big]$.  We note that the original constraint C3 is satisfied whenever $\mbox{\textoverline{C3}}$ is  satisfied. Besides, the new  constraint $\mbox{\textoverline{C3}}$ is not only an affine function  with respect to the optimization variables, but also involves only a finite number of constraints. In particular, $\mbox{\textoverline{C3}}$ can be easily handled by standard convex program solvers.}


Next, we handle non-convex constraint C4 by introducing the following proposition for simplifying the considered optimization problem.
  \begin{Prop}\label{prop:det_constraint} For $R_{\mathrm{Eav}_{j,k}}>0$, the following implication holds for  constraint C4:
\begin{eqnarray}\label{eqn:C4:LMI}
\mbox{C4}\Longrightarrow \mbox{\textovertide{C4}: }\hspace{-8mm}&&  \max_{\norm{\Delta\mathbf{G}_j}_F\in {\Psi}_j}\,\,  \mathbf{G}_j^H
\mathbf{W}_{1,k} \mathbf{G}_j\\
&\preceq& \xi_{\mathrm{Eav}_{j,k}}\mathbf{\Sigma}_{j},\forall j\in\{1,\ldots,J\}, \forall k\in\{1,\ldots,K\},\notag
\end{eqnarray}
\end{Prop}
where $\xi_{\mathrm{Eav}_{j,k}}=2^{R_{\mathrm{Eav}_{j,k}}}-1$ is an auxiliary constant with  $\xi_{\mathrm{Eav}_{j,k}}> 0$
for $R_{\mathrm{Eav}_{j,k}}> 0$. We note that constraint $\mbox{\textovertide{C4}}$ is equivalent to constraint $\mbox{{C4}}$ if $\Rank(\mathbf{W}_{1,k})\le 1,\forall k$.

 \emph{\,Proof:} Please refer to Appendix A.

Although constraint $\mbox{\textovertide{C4}}$ is less complex compared to  $\mbox{{C4}}$, there are still infinitely many LMI constraints to satisfy  $\mbox{\textovertide{C4}}$ for $\norm{\Delta\mathbf{G}_j}_F\in {\Psi}_j$. Hence, we adopt the following Lemma to further simplify $\mbox{\textovertide{C4}}$:

\begin{Lem}[Robust Quadratic Matrix Inequalities \cite{JR:extended_S_procedure}] Let a quadratic matrix function $f(\mathbf{X})$ be defined as
\label{Lemma:S_matrix}
\begin{eqnarray}
f(\mathbf{X})=\mathbf{X}^H\mathbf{A}\mathbf{X}+\mathbf{X}^H\mathbf{B}+\mathbf{B}^H\mathbf{X} +\mathbf{C},
\end{eqnarray}
 where $\mathbf{X},\mathbf{A}, \mathbf{B}$, and $\mathbf{C}$ are arbitrary matrices with appropriate dimensions.  Then, the following two statements are equivalent: \begin{eqnarray}\label{eqm:extended_S_lemma1}
&&f(\mathbf{X})\succeq \mathbf{0},\forall \mathbf{X}\in\Big\{\mathbf{X}\mid \Tr(\mathbf{D}\mathbf{X}\mathbf{X}^H)\le 1\Big\}\notag\\
&&\Longleftrightarrow\label{eqm:extended_S_lemma2}\begin{bmatrix}
       \mathbf{C} & \mathbf{B}^H          \\
       \mathbf{B} & \mathbf{A}           \\
           \end{bmatrix} -\delta\begin{bmatrix}
       \mathbf{I} & \mathbf{0}          \\
       \mathbf{0} & -\mathbf{D}           \\
           \end{bmatrix}          \succeq \mathbf{0},  \,\,\mbox{if } \exists\delta\ge 0,
\end{eqnarray}
for matrix $\mathbf{D}\succeq \zero$ and $\delta$ is an auxiliary constant.
\end{Lem}
We note that  Lemma \ref{Lemma:S_matrix} has been adopted in the literature before for resource allocation algorithm design with imperfect CSI \cite{JR:MIMO_Robust_CR}. By applying Lemma \ref{Lemma:S_matrix} to (\ref{eqn:C4:LMI}) and following similar steps as in \cite{JR:Ken_AN_PHY}, i.e., setting $\mathbf{X}=\Delta\mathbf{G}_j$, $\mathbf{A}= \xi_{\mathrm{Eav}_{j,k}}\mathbf{V}-\mathbf{W}_{1,k}$, $\mathbf{B}= (\xi_{\mathrm{Eav}_{j,k}}\mathbf{V}-\mathbf{W}_{1,k})\mathbf{\widehat G}_j, \mathbf{C}=\xi_{\mathrm{Eav}_{j,k}}\mathbf{I}_{N_\mathrm{R}}\sigma_{\mathrm{PU}_j}^2+{\mathbf{\widehat G}}^H_j(\xi_{\mathrm{Eav}_{j,k}}\mathbf{V}- \mathbf{W}_{1,k}){\mathbf{\widehat G}}^H_j$, and $\mathbf{D}=\frac{\mathbf{I}_{N_{\mathrm{T}}}}{\varepsilon_j^2}$ in (\ref{eqm:extended_S_lemma1}), we obtain

\begin{eqnarray}\label{eqn:LMI_C4}
&&\mbox{C4}\Longrightarrow \mbox{\textovertide{C4}}\Longleftrightarrow\mbox{\textoverline{C4}: } \mathbf{S}_{\mathrm{\overline{C4}}_{k,j}}({\mathbf{W}}_{1,k},\mathbf{V}, \delta_{k,j})\notag\\
=&&\hspace*{-5mm}\notag
\begin{bmatrix}
        \mathbf{F}_{j,k}& \xi_{\mathrm{Eav}_{j,k}}{\mathbf{\widehat G}}^H_j\mathbf{V}     \\
       \xi_{\mathrm{Eav}_{j,k}}\mathbf{V} {\mathbf{\widehat G}}_j \ &  \xi_{\mathrm{Eav}_{j,k}}\mathbf{V}+\frac{\delta_{k,j}}{\varepsilon_j^2} \mathbf{I}_{N_{\mathrm{T}}}
           \end{bmatrix}\notag\\
           -&&\mathbf{R}_j^H\mathbf{W}_{1,k}\mathbf{R}_j\succeq \mathbf{0},\forall k, j,
\end{eqnarray}
where $\mathbf{R}_j=\big[{\mathbf{\widehat G}}_j, \,\mathbf{I}_{N_\mathrm{T}}\big]$, $\mathbf{F}_{i,k}=(\xi_{\mathrm{Eav}_{j,k}}\sigma_{\mathrm{PU}_j}^2-\delta_{k,j})\mathbf{I}_{N_\mathrm{R}}+      \xi_{\mathrm{Eav}_{j,k}}{\mathbf{\widehat G}}^H_j \mathbf{V}{\mathbf{\widehat G}}_j $, and $\delta_{k,j}$ is an auxiliary optimization variable. Besides, C4 is equivalent to  $\mbox{\textoverline{C4}} $ when $\Rank(\mathbf{W}_{1,k})\le 1$.

Now,  we replace constraints $\mbox{C3}$ and $\mbox{C4}$ with constraints $\mbox{\textoverline{C3}}$ and $\mbox{\textoverline{C4}}$, respectively.   Hence, the new optimization problem can be written as
\begin{eqnarray}
\label{eqn:SDP_transformed}&&\hspace*{-10mm} \underset{\mathbf{W}_{l,k},\mathbf{V}\in \mathbb{H}^{N_\mathrm{T}},\omega_j,\delta_{k,j}
}{\mino} \,\, \sum_{k=1}^K\sum_{l=1}^{L_k}\Tr(\mathbf{W}_{l,k})+\Tr(\mathbf{V})\nonumber\\
\notag \mbox{s.t. } \hspace*{-1mm}&&\hspace*{25mm}\mbox{C1, C2, C5, C6,}\notag \\
\hspace*{6mm}&&\hspace*{-5mm}\notag  \mbox{\textoverline{C3}: }\mathbf{S}_{\mathrm{\overline{C3}}_j}({\mathbf{W}}_{l,k},{\mathbf{V}},\omega_j)\succeq \mathbf{0},\forall j\notag\\
\hspace*{6mm}&&\hspace*{-5mm}\notag \mbox{\textoverline{C4}: } \mathbf{S}_{\mathrm{\overline{C4}}_{k,j}}({\mathbf{W}}_{1,k},\mathbf{V}, \delta_{k,j})\succeq \mathbf{0}\notag,\,\,\,\forall j,\forall k,\\
\hspace*{6mm}&&\hspace*{-5mm}\mbox{C7:}\,\Rank(\mathbf{W}_{l,k})\le 1,\forall k,l, \,\,\mbox{C8: }\omega_j\ge 0,\forall j,\notag\\
\hspace*{6mm}&&\hspace*{-5mm}\mbox{C9:  }\delta_{k,j}\ge 0,\forall k,j,
\end{eqnarray}
where $\omega_j$ and $\delta_{k,j}$ in C8 and C9 are connected to the LMI
constraints in (\ref{eqn:LMI_C3}) and (\ref{eqn:LMI_C4}), respectively. Since optimization problems  (\ref{eqn:SDP_transformed}) and (\ref{eqn:SDP}) share the same optimal solution, we focus on the design of the optimal resource allocation policy for the problem in (\ref{eqn:SDP_transformed})  in the sequel.

{
We note that constraints  $\mbox{\textoverline{C3}}$ and $\mbox{\textoverline{C4}}$ are jointly convex with respect to the optimization variables. The only remaining obstacle in solving (\ref{eqn:SDP_transformed}) is the combinatorial rank constraint in $\mbox{C7}$. Hence, we adopt the SDP relaxation approach by relaxing constraint $\mbox{C7: }\Rank(\mathbf{W}_{l,k})\le 1$, i.e., we remove C7 from the problem formulation. Then, the considered problem becomes a convex SDP which can be solved efficiently by numerical solvers such as CVX \cite{website:CVX}.  However, removing constraint $\mbox{C7 }$  results in a larger feasible solution set. Hence, in general, the optimal objective value of the relaxed problem  of (\ref{eqn:SDP_transformed}) may be smaller than the optimal objective value of  (\ref{eqn:SDP}). If the solution $\mathbf{W}_{l,k}$ of the relaxed problem is a rank-one matrix, this  is also the optimal solution of the original problem in (\ref{eqn:SDP}) and the adopted SDP relaxation is tight. Subsequently, the optimal $\mathbf{w}_{l,k}$ can be obtained by performing eigenvalue decomposition of $\mathbf{W}_{l,k}$ and selecting the principal eigenvector as the beamforming vector. Unfortunately, in general the constraint relaxation may not be tight and $\Rank(\mathbf{W}_{l,k})>1$ may occur.  In the following, we propose a method for constructing an optimal solution of the relaxed version of (\ref{eqn:SDP_transformed}) with a rank-one matrix $\mathbf{W}_{l,k},\forall k,l$.}
\newcounter{mytempeqncnt}
 \begin{figure*}[!t]\setcounter{mytempeqncnt}{\value{equation}}
\setcounter{equation}{27}
\begin{eqnarray}
\label{eqn:sub-optimal_1}&&\hspace*{-10mm} \underset{\mathbf{V}\in \mathbb{H}^{N_\mathrm{T}},P_{l,k},\omega_j,\delta_{k,j}
}{\mino} \,\, \sum_{k=1}^K\sum_{l=1}^{L_k}\Tr(P_{l,k}\mathbf{ W}^\mathrm{sub}_{l,k})+\Tr(\mathbf{V})\nonumber\\
\notag \mbox{s.t. } \hspace*{-1mm}&&\hspace*{-5mm}\mbox{C1: }\notag \frac{\Tr(\mathbf{H}_k P_{l,k}\mathbf{ W}^\mathrm{sub}_{l,k})}{ \Tr\Big(\mathbf{H}_k \big(\overset{K}{\underset{n\ne k}{\sum}}\overset{L_n}{\underset{l=1}{\sum}}P_{l,n}\mathbf{ W}^\mathrm{sub}_{l,n}+\overset{L_k}{\underset{m=l+1}{\sum}}P_{m,k}\mathbf{ W}^\mathrm{sub}_{m,k} \big)\Big)+\Tr(\mathbf{H}_k\mathbf{V})+\sigma_{\mathrm{s}_k}^2} \ge \Gamma_{\mathrm{req}_{l,k}}, \forall l,\forall k, \\
\hspace*{-1mm}&&\hspace*{-5mm}\mbox{C2: } \notag\frac{\Tr(\mathbf{H}_t P_{1,k}\mathbf{ W}^\mathrm{sub}_{1,k})}{ \Tr\Big(\hspace*{-0.5mm}\mathbf{H}_t\big(\mathbf{V}+\overset{K}{\underset{n\ne t}{\underset{n\ne k}{\sum}}}\overset{L_n}{\underset{l=1}{\sum}}P_{l,n}\mathbf{ W}^\mathrm{sub}_{l,n}\hspace*{-0.5mm}+\hspace*{-0.5mm}\overset{L_k}{\underset{m=2}{\sum}}P_{m,k}\mathbf{ W}^\mathrm{sub}_{m,k}\big)\hspace*{-0.5mm}\Big)+\sigma_{\mathrm{s}_k}^2} \le \Gamma_{\mathrm{tol}}, \forall t\ne k, t\in\{1,\ldots,K\} , \\
\hspace*{6mm}&&\hspace*{-5mm}\notag  \mbox{\textoverline{C3}: }\mathbf{S}_{\mathrm{\overline{C3}}_j}(P_{l,k}\mathbf{ W}^\mathrm{sub}_{l,k},{\mathbf{V}},\omega_j)\succeq \mathbf{0},\forall j,\quad\mbox{\textoverline{C4}: } \mathbf{S}_{\mathrm{\overline{C4}}_{k,j}}(P_{1,k}\mathbf{ W}^\mathrm{sub}_{1,k},\mathbf{V}, \delta_{k,j})\succeq \mathbf{0}\notag,\,\,\,\forall j,\forall k,\\
\hspace*{6mm}&&\hspace*{-5mm}\mbox{C5:}\,\, \mathbf{V}\succeq \mathbf{0},\,\,\,\mbox{C6: }P_{l,k}\ge 0,\,\,\,\mbox{C8: }\omega_j\ge 0,\forall j,\,\,\,\mbox{C9:  }\delta_{k,j}\ge 0,\forall k,j.
\end{eqnarray}\hrulefill\setcounter{equation}{28}
\end{figure*}

\subsection{Optimality Condition for SDP Relaxation}
 In this subsection, we first reveal the tightness of the proposed SDP relaxation. The existence of a rank-one solution  matrix $\mathbf{W}_{l,k}$ for the  relaxed SDP version of (\ref{eqn:SDP_transformed}) is summarized in
the following theorem which is based on  \cite[Proposition 4.1]{JR:rui_zhang}.

\begin{Thm}\label{prop1}Suppose the optimal solution of the SDP relaxed version of (\ref{eqn:SDP_transformed}) is denoted by $\{\mathbf{W}_{l,k}^*,\mathbf{V}^*,\omega_j^*,$ $\delta_{k,j}^*\}$ and  $\exists k,l: \Rank(\mathbf{W}^*_{l,k})>1$. Then, there exists a feasible optimal solution of the SDP relaxed version of (\ref{eqn:SDP_transformed}), denoted by  $\mathbf{\widetilde \Lambda}\triangleq\{\mathbf{\widetilde W}_{l,k},\mathbf{\widetilde V},\widetilde\omega_j,\widetilde\delta_{k,j}\}$, with a rank-one matrix $\mathbf{\widetilde W}_{l,k}$, i.e.,  $\Rank(\mathbf{\widetilde W}_{l,k})=1$. This optimal solution can be obtained by construction.
\end{Thm}

\begin{proof}
Please refer to Appendix B for the proof of Theorem \ref{prop1} and the method for constructing the optimal solution.
\end{proof}
In other words,  the optimal solution of the SDP relaxed version of (\ref{eqn:SDP_transformed}) is  a rank-one beamforming matrix  $\mathbf{\widetilde W}_{l,k},\forall l,k$, by construction. Thus, constraint C4 is equivalent to $\mbox{\textovertide{C4}}$. More importantly,  the global optimum of (\ref{eqn:SDP}) can be obtained despite the adopted SDP relaxation.
\subsection{Suboptimal Resource Allocation Schemes}
The construction of  the optimal solution $\mathbf{\widetilde \Lambda}$ with $\Rank(\mathbf{\widetilde W}_{l,k})=1$ requires the optimal solution of the  dual version of the relaxed problem of (\ref{eqn:SDP_transformed}), cf. variable $\mathbf{Y}^*_{l,k}$ in (\ref{eqn:KKT-gradient-equivalent}) in Appendix B.
{
However, the solution of the dual problem   may not be provided by some numerical solvers and thus the construction of a rank-one  matrix $\mathbf{\widetilde W}_{l,k}$ is not possible in such cases.  In the following, we propose two suboptimal resource allocation schemes based on the solution of the primal problem of the relaxed version of (\ref{eqn:SDP_transformed}) which do not require the solution of the dual problem.}

\subsubsection{Suboptimal Resource Allocation Scheme 1}
 A suboptimal resource allocation scheme is proposed which is based on the  solution of the relaxed version of (\ref{eqn:SDP_transformed}). We first solve (\ref{eqn:SDP_transformed}) by SDP relaxation.  The global optimal solution of  (\ref{eqn:SDP_transformed}) is found if the obtained solution $\mathbf{W}_{l,k}^*$ is a  rank-one matrix.  Otherwise, we  construct a suboptimal solution set $ \mathbf{ W}^\mathrm{sub}_{l,k} = \mathbf{w}^\mathrm{sub}_{l,k}(\mathbf{w}^\mathrm{sub}_{l,k})^H$, where $\mathbf{w}^\mathrm{sub}_{l,k}$ is the eigenvector corresponding to the principal eigenvalue of matrix $\mathbf{ W}_{l,k}^*$.  Then, we define a scalar optimization variable $P_{l,k}$ which controls the power of the suboptimal beamforming matrix of layer $l$ for secondary receiver $k$. Subsequently, a new optimization problem is formulated as \eqref{eqn:sub-optimal_1} on the top of this page. It can be shown that the above optimization problem is jointly convex with respect to the optimization variables and thus can be solved by using efficient numerical solvers. Besides, the solution of (\ref{eqn:sub-optimal_1}) also satisfies the constraints of (\ref{eqn:SDP}). In other words, the solution of (\ref{eqn:sub-optimal_1})  serves as a suboptimal solution for (\ref{eqn:SDP}).

\subsubsection{Suboptimal Resource  Allocation Scheme 2}
The second proposed suboptimal resource allocation scheme  adopts a similar approach to solve the problem  as   suboptimal resource  allocation scheme 1, except for the choice of the suboptimal beamforming matrix $\mathbf{ W}^\mathrm{sub}_{l,k}$ when  $\Rank( \mathbf{ W}^\mathrm{sub}_{l,k})>1$. For scheme 2, the choice of beamforming matrix  $\mathbf{ W}^\mathrm{sub}_{l,k}$  is based on
the rank-one Gaussian randomization scheme  \cite{JR:Gaussian_randomization}.  Specifically, we calculate the eigenvalue decomposition of $\mathbf{ W}_{l,k}=\mathbf{U}_{l,k}\mathbf{\Theta}_{l,k}\mathbf{U}^H_{l,k}$, where $\mathbf{U}_{l,k}$ and $\mathbf{\Theta}_{l,k}$  are an $N_\mathrm{T}\times N_\mathrm{T}$ unitary matrix and a diagonal matrix, respectively. Then, we adopt the suboptimal beamforming vector $\mathbf{w}^\mathrm{sub}_{l,k}=\mathbf{U}_{l,k}\mathbf{\Theta}_{l,k}^{1/2}\mathbf{q}_{l,k},  \mathbf{W}^\mathrm{sub}_{l,k}=P_{l,k}\mathbf{w}^\mathrm{sub}_{l,k}(\mathbf{w}^\mathrm{sub}_{l,k})^H$, where $\mathbf{q}_{l,k}\in {\mathbb C}^{N_{\mathrm{T}}\times 1}$ and $\mathbf{q}_{l,k}\sim {\cal CN}(\mathbf{0}, \mathbf{I}_{N_{\mathrm{T}}})$. Subsequently, we follow the same approach as in (\ref{eqn:sub-optimal_1}) for optimizing $\{\mathbf{V},P_{l,k},\omega_j,\delta_{k,j}\}$ and obtain a suboptimal rank-one solution $P_{l,k} \mathbf{ W}^\mathrm{sub}_{l,k}$.  Furthermore,  we can execute  scheme 2 repeatedly for different  realizations of the Gaussian distributed random vector $\mathbf{q}_{l,k}$ such that the performance of scheme 2 can be improved by selecting the best $ \mathbf{ w}^\mathrm{sub}_{l,k}=\mathbf{U}_{l,k}\mathbf{\Theta}_{l,k}^{1/2}\mathbf{q}_{l,k}$ over different trials at the expense of a higher computation complexity.

\subsection{Computational Complexity}
{
In this section, we study the computational complexity of the proposed optimal and the two suboptimal algorithms.
 An upper bound  for the computational complexity of the optimal algorithm is given by \cite{book:interior_point_complexity}:
   \begin{eqnarray}
\hspace*{-6mm}&&\Delta_{\mathrm{complexity}}^{\mathrm{Opt}}=2\times\Delta_{\mathrm{complexity}}^{\mathrm{SDP}},\\
\label{eqn:big_O}
\hspace*{-6mm}&&\Delta_{\mathrm{complexity}}^{\mathrm{SDP}}=\bigo\Bigg(\Big(\sqrt{N_{\mathrm{T}}(LK\hspace*{-0.5mm}+\hspace*{-0.5mm}1)}
\log(\frac{1}{\delta})\Big)
\Big((N_{\mathrm{T}}(LK\hspace*{-0.5mm}+\hspace*{-0.5mm}1))^3\notag\\
\hspace*{-6mm}&&(KL\hspace*{-0.5mm}+\hspace*{-0.5mm}K^2\hspace*{-0.5mm}+\hspace*{-0.5mm}J(K\hspace*{-0.5mm}+\hspace*{-0.5mm}1)) \hspace*{-0.5mm}+\hspace*{-0.5mm}(N_{\mathrm{T}}(LK\hspace*{-0.5mm}+\hspace*{-0.5mm}1))^2\notag\\
\hspace*{-6mm}&&(KL\hspace*{-0.5mm}+\hspace*{-0.5mm}K^2\hspace*{-0.5mm}+\hspace*{-0.5mm}J(K\hspace*{-0.5mm}+\hspace*{-0.5mm}1))^2 \hspace*{-0.5mm}+\hspace*{-0.5mm}(KL\hspace*{-0.5mm}+\hspace*{-0.5mm}K^2\hspace*{-0.5mm}+\hspace*{-0.5mm}J(K\hspace*{-0.5mm}+\hspace*{-0.5mm}1))^3\Big) \Bigg)
\end{eqnarray}
 for a given solution accuracy $\delta>0$, since at most two SDPs are solved. In \eqref{eqn:big_O},   $\bigo(\cdot)$ represents for the big-O notation. On the other hand,  the computational complexity upper bound of suboptimal scheme 1 is given by
   \begin{eqnarray}
2\times\Delta_{\mathrm{complexity}}^{\mathrm{SDP}}
\end{eqnarray}
 since two SDPs have to be solved in this case. Suboptimal algorithm 2 adopts a similar approach to solve the problem as suboptimal scheme 1. The only difference is the multiple attempts in generating a Gaussian distributed beamforming vector for improving the system performance. Hence, the computational complexity upper bound for  suboptimal scheme 2 is given by
\begin{eqnarray}
(N_{\mathrm{Tries}}+1)\times\Delta_{\mathrm{complexity}}^{\mathrm{SDP}},
\end{eqnarray}
where $N_{\mathrm{Tries}}$ is the number of tries  in generating a Gaussian distributed beamforming vector. We note that the proposed optimal and suboptimal algorithms have polynomial time computational complexity. Such algorithms are considered to be fast algorithms in the literature  \cite[Chapter 34]{book:polynoimal} and are desirable for real time implementation.}

\section{Results}\label{sect:simulation}
In this section, we study the system performance of the proposed resource allocation scheme via simulations.  There are $K$ secondary receivers and $J$ primary receivers,  which are uniformly distributed in the range between a reference distance of $30$ meters and the maximum cell radius of $500$ meters.  We assume that there is always one \emph{premium secondary receiver} and the secondary transmitter is required to guarantee the SINR  of all video layers for  this receiver. On the contrary, the transmitter   guarantees only  the SINR of the first layer for the remaining $K-1$ \emph{regular receivers}.  We assume that the video signal of each secondary receiver is encoded into two layers.  { As is commonly done in the literature \cite{CN:two_layer}\nocite{CN:two_layer_1}--\cite{JR:two_layer}, we limit our case study to a single enhancement layer, since each additional enhancement layer increases the delay. } For the sake of illustration, the minimum required  SINR of the first layer and the second layer are given by $\Gamma_{\mathrm{Base}}$ and $\Gamma_{\mathrm{Base}}+3$ dB, respectively.  Also, we solve the optimization problem in (\ref{eqn:cross-layer}) via SDP relaxation and obtain the  average system performance by averaging over  different channel realizations.  We assume that the primary transmitter is equipped with $N_{\mathrm{P}_{\mathrm{T}}}=8$ antennas which serve all primary receivers simultaneously. The primary transmitter is located $500$ meters away from the secondary transmitter and transmits with a power\footnote{We assume that the primary transmitter has a lower
maximum transmit power budget compared to the secondary transmitter. In fact, both the primary transmitters and  receivers are equipped with multiple antennas which facilities power efficient data communication in the primary network.}
 of $5$ dBm. Because of path loss and channel fading, different secondary receivers experience different interference powers from the primary transmitter. In the sequel, we define the normalized maximum  channel estimation error of primary receiver $j$  as  $\sigma_{\mathrm{PU}_j}^2=\frac{\varepsilon^2_j}{\norm{\mathbf{G}_j}^2_F}$  with $\sigma_{\mathrm{PU}_a}^2=\sigma_{\mathrm{PU}_b}^2,\forall a, b\in\{1,\ldots,J\}$.  Unless specified otherwise, we assume  a normalized maximum  channel estimation error of  $\sigma_{\mathrm{PU}_j}^2=0.05,\forall j$ for primary receiver $j$  and there are $N_{\mathrm{P}_\mathrm{R}}=2$ receive antennas at each primary receiver. Besides, the maximum tolerable interference power at the primary receivers is set to $P_{\mathrm{I}_j}=-110.35$ dBm, $\forall j\in\{1,\ldots,J\}$. The  parameters adopted for our simulation are summarized in Table \ref{tab:feedback}.

\begin{table}[t]\caption{System parameters}\label{tab:feedback} \centering

\begin{tabular}{|L|p{3.5cm}|}\hline

Carrier center frequency & $2.6$ GHz  \\
  \hline
Small-scale fading distribution & Rayleigh fading\\
\hline
Large-scale fading model  & Non-line-of-sight, urban micro scenario, 3GPP \cite{3Gpp:09} \\
 \hline
Cell radius  & $500$ meters \\
    \hline
    Transceiver antenna gain & $0$ dBi \\
    \hline
    Thermal noise power, ${\cal E} \{\abs{n_{\mathrm{s}_k}}^2\}$, $\sigma_{\mathrm{PU}_j}^2$ &  $-107.35$ dBm   \\
    \hline
     Maximum tolerable received interference power at  primary receiver $j$, $P_{\mathrm{I}_j}$ &  $-110.35$ dBm    \\
    \hline
    Minimum requirement on the SINR of layers  $[\Gamma_{\mathrm{req}_1},\,\,\Gamma_{\mathrm{req}_2}]$ & $[\Gamma_{\mathrm{Base}},\,\, \Gamma_{\mathrm{Base}}+3 ]$ dB   \\
        \hline
         Maximum tolerable SINR for information decoding at unintended primary receivers,  $\Gamma_{\mathrm{tol}}$ & $0$ dB \\
                 \hline
         Maximum tolerable data rate at primary receiver,  $R_{\mathrm{Eav}_{j,k}}$ & $1$ bit/s/Hz \\
                \hline
                   Transmit power of primary transmitter   &  $5$ dBm \\
                \hline
\end{tabular}
\end{table}

\subsection{Average Total Transmit Power versus Minimum Required SINR}
Figure \ref{fig:p_SNR} depicts the  average total transmit power versus the minimum required SINR of the base layer, $\Gamma_{\mathrm{Base}}$,    for $N_{\mathrm{T}}=8$ transmit antennas,  $K=2$ secondary receivers, $J=2$ primary receivers, and different resource  allocation schemes.
It can be observed that  the average total transmit power for the proposed schemes
is a monotonically increasing function with respect to the minimum required SINR of the base layer.  Clearly, the transmitter has to allocate more power to the information signal
  as the  SINR requirement gets more stringent.
    Besides, the two   proposed suboptimal resource allocation schemes approach the optimal performance.    In fact,  the proposed suboptimal schemes exploit the possibility of achieving the global optimal solution via SDP relaxation. { We note that extensive simulations (not shown
here) have revealed that, for the considered scenarios, the percentage of rank-one solution of the SDP relaxed problem in \eqref{eqn:SDP_transformed} ranges from $75\%$ to $100\%$. Nevertheless, the proposed optimal algorithm is always able to reconstruct the optimal solution by utilizing the solution of the dual problem.}

\begin{figure}[t]
 \centering
\includegraphics[width=3.5in]{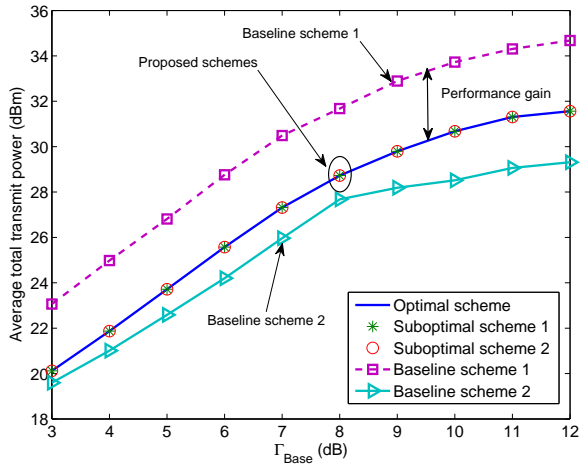}
\caption{Average total transmit power (dBm) versus minimum required SINR of the base layer, $\Gamma_{\mathrm{Base}}$.   } \label{fig:p_SNR}
\end{figure}

For comparison, Figure
\ref{fig:p_SNR} also contains results for the average total transmit power of two baseline resource allocation schemes. For baseline scheme 1, we adopt single-layer transmission for  delivering the multiuser video signals. In particular,
we solve the corresponding robust optimization problem with respect to $\{\mathbf{W}_{l,k},\mathbf{V},\omega_j,\delta_{k,j}\}$ subject to constraints C1 -- C9 via
 SDP relaxation. The minimum required SINR  for decoding the single-layer video information  at the secondary receivers for baseline scheme 1 is set to\footnote{We note that the actual data rate for multi-layer and single-layer transmission depends heavily on the adopted video coding algorithm. In order to isolate the performance study from the video coding implementation details, we adopt the information theoretic approach which focuses on the channel dependent achievable data rate.  } $\Gamma_{\mathrm{req}_k}^{\mathrm{Single}}=2^{\sum_{l=1}^{L_k}\log_2(1+\Gamma_{\mathrm{req}_{l,k}})}-1$. In baseline scheme 2, we consider a naive layered video transmission. Specifically, the secondary transmitter treats the estimated CSI of the primary receivers as perfect CSI and exploits it for resource allocation. In other words, robustness against CSI errors is not provided by baseline scheme 2.   It can be observed that baseline scheme 1 requires a higher total average power  compared to the proposed resource allocation schemes. This can be attributed to the fact that single-layer transmission  does not posses the   \emph{self-protecting} structure for providing secure communication  that layered transmission has. As a result,
 a higher transmit power is required in baseline scheme 1 to ensure secure video delivery. On the other hand, it is expected that for  baseline scheme 2, the average transmit power is lower than that of the proposed scheme.  This is due to the fact that the secondary transmitter assumes the available CSI is perfect and transmits with insufficient power for providing secure communication. The next sections will show that baseline scheme 2 cannot meet the QoS requirements regarding communication security and interference leakage to the primary network.
\begin{figure}[t] \centering
\includegraphics[width=3.5in]{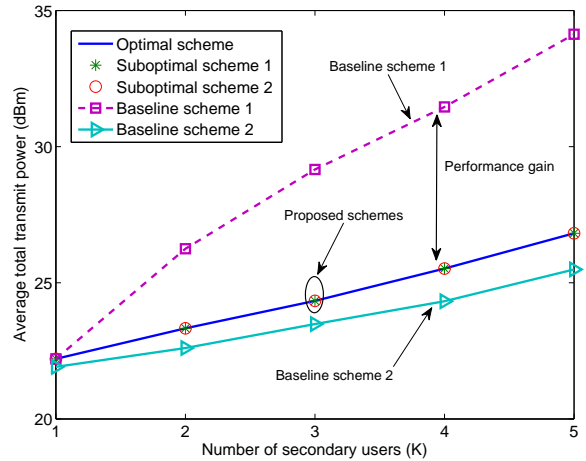}
\caption{Average total transmit power (dBm) versus the number of secondary receivers. } \label{fig:pt_k}

\end{figure}

\subsection{Average Total Transmit Power versus Number of Secondary Receivers}
Figure \ref{fig:pt_k} illustrates the average total transmit power versus the
number of secondary receivers for a minimum required SINR of the base layer of $\Gamma_{\mathrm{Base}}=5$ dB, $J=1$ primary receiver,  $N_{\mathrm{T}}=8$ transmit antennas,  and  different resource allocation schemes.
 It can be seen that the average total transmit power
increases  with the number of secondary receivers for all resource allocation schemes.  In fact,  the requirement of secure communication becomes more difficult to meet if there are more secondary receivers in the system. Besides, more degrees of freedom are utilized for reducing the mutual interference between the secondary receivers which leads to a less efficient power allocation. Hence, a higher total transmit power is required to meet the target QoS.

 On the other hand,  the two   proposed suboptimal resource allocation schemes achieve
a similar performance as the optimal resource allocation scheme. Also, the proposed schemes provide substantial
power savings compared to  baseline scheme 1 for $K>1$ due to the adopted layered transmission.  In particular, the performance gap between the proposed schemes and baseline scheme 1 increases with increasing number of secondary receivers.  In other words, layered transmission is effective for reducing the transmit power in multi-receiver environments with secrecy constraints, due to the self-protecting property. As for baseline scheme 2, although it consumes less transmit power compared to the optimal scheme, it cannot guarantee the QoS in communication secrecy and interference to the primary receivers, cf. Figures \ref{fig:cap_SINR} -- \ref{fig:interference_users}.
\begin{figure}[t]
 \centering
\includegraphics[width=3.5in]{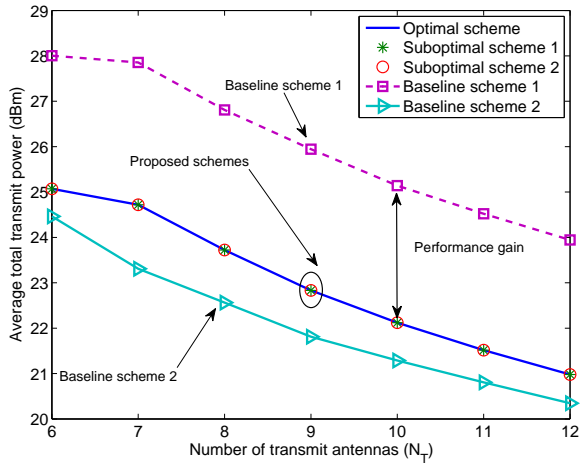}
\caption{Average total transmit power (dBm) versus the number of transmit antennas, $N_{\mathrm{T}}$.  } \label{fig:pt_nt}
\end{figure}

\subsection{Average Total Transmit Power versus Number of Antennas}
Figure \ref{fig:pt_nt} shows the average total transmit power versus the
number of transmit antennas, $N_{\mathrm{T}}$,  for a minimum required SINR of the base layer of $\Gamma_{\mathrm{Base}}=5$ dB, $J=2$ primary receivers, $K=2$ secondary receivers, and  different resource allocation schemes.
 It is expected that the average total transmit power
decreases for all resource allocation schemes with increasing number of transmit antennas. This is because extra degrees of freedom
 can be exploited for resource allocation when more antennas are available at the transmitter. Specifically, with more antennas, the direction of beamforming matrix $\mathbf{W}_{l,k}$ can be more accurately steered towards the secondary receivers which reduces both the power consumption at the secondary transmitter and the power leakage to the primary receivers.
On the other hand,   the proposed schemes  provide  substantial power savings compared to  baseline scheme 1 for all considered
scenarios because of the adopted layered transmission. Besides, baseline scheme 2 consumes less transmit power compared to the optimal scheme again. Although baseline scheme 2 can exploit the extra degrees of freedom offered by the increasing number of antennas, it is unable to protect the primary receivers from interference and cannot guarantee communication security due to the imperfection of the CSI, cf. Figures \ref{fig:cap_SINR} -- \ref{fig:interference_users}.

\begin{figure}[t]
 \centering
\includegraphics[width=3.5in]{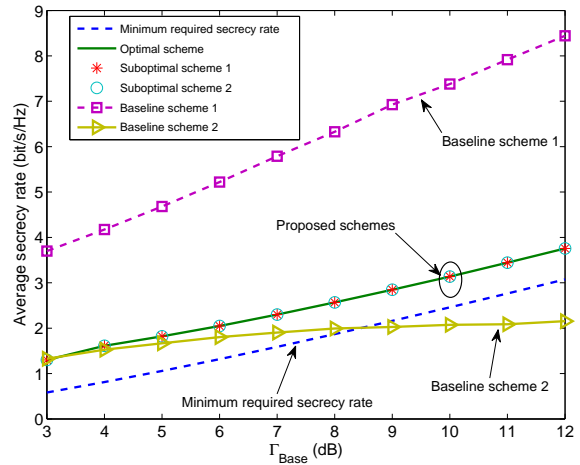}
\caption{Average secrecy rate (bit/s/Hz) of the base layer versus the minimum required SINR of the base layer, $\Gamma_{\mathrm{Base}}$.  } \label{fig:cap_SINR}
\end{figure}
\subsection{Average Secrecy Rate}
Figure \ref{fig:cap_SINR} depicts the average secrecy rate  of the base layer versus  the minimum required SINR of the base layer for $N_{\mathrm{T}}=8$ transmit antennas,  $K=2$ secondary receivers, $J=2$ primary receivers, and different resource  allocation schemes. Despite the imperfection of the CSI, the proposed  optimal resource allocation scheme and the two suboptimal resource allocation schemes are able to guarantee the minimum secrecy rate defined by constraints C1, C2, and C4 in every time instant, because of the adopted robust optimization framework. On the other hand, baseline scheme 1 achieves an exceedingly high average secrecy rate since the entire video information is encoded in the first layer. The superior secrecy rate performance of  baseline scheme 1  comes at the expense of an exceedingly high transmit power, cf. Figure \ref{fig:p_SNR}. In the low  $\Gamma_{\mathrm{Base}}$ regime,  even though baseline  scheme 2 is able to meet the  minimum secrecy rate requirement on average, we  emphasize that baseline scheme 2 is unable to fulfill the requirement for all channel realizations, i.e., secure communication is not ensured.  Besides, in the high  $\Gamma_{\mathrm{Base}}$ regime, in contrast to the proposed schemes, baseline scheme 2 cannot even satisfy the minimum secrecy rate requirement on average\footnote{ We note that the performance of  baseline schemes without artificial noise generation is not shown in the paper since a feasible solution cannot be found under the adopted simulation parameters.}.

\begin{figure}[t]
 \centering
 \includegraphics[width=3.5in]{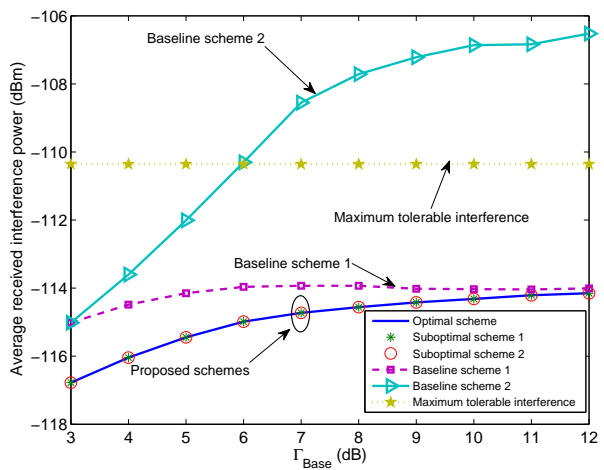}
\caption{Average received interference power (dBm) at each primary receiver versus the minimum required SINR of the base layer, $\Gamma_{\mathrm{Base}}$.} \label{fig:interference_SINR}
\end{figure}

\subsection{Average Interference Power}
Figure \ref{fig:interference_SINR} depicts the average received interference power at each primary receiver  versus the minimum required SINR of the base layer $\Gamma_{\mathrm{Base}}$,    for $N_{\mathrm{T}}=8$ transmit antennas, $K=2$ secondary receivers, $J=2$ primary receivers,  and different resource  allocation schemes. As can be observed, the proposed  optimal resource allocation scheme and the two suboptimal resource allocation schemes are able to control their transmit power such that the received interference powers at the primary receivers are below the maximum tolerable interference threshold. Similar results can be observed for baseline scheme 1 as robust optimization is also adopted in this case. As for baseline scheme 2,  although  the average interference received by each primary receiver is below the maximum tolerable threshold for $\Gamma_{\mathrm{Base}}\le 6$ dB, baseline scheme 2 cannot meet the interference requirement for all channel realizations. Besides, as the value of $\Gamma_{\mathrm{Base}}$ increases, the received interference power at each primary receiver increases significantly compared to the proposed schemes. For high values of $\Gamma_{\mathrm{Base}}$,  even the average received interference at each primary receiver for baseline scheme 2 exceeds the maximum tolerable interference limit.

\begin{figure}[t]
 \centering
\includegraphics[width=3.5in]{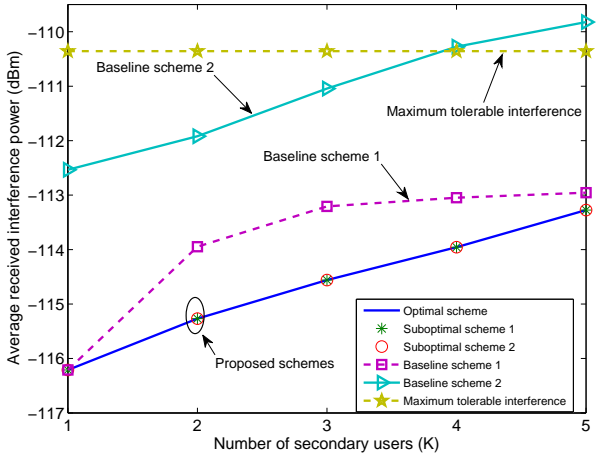}
\caption{Average received interference power (dBm) at each primary receiver versus the number of secondary receivers, $K$. } \label{fig:interference_users}
\end{figure}

Figure \ref{fig:interference_users} shows the average received interference power at each primary receiver  versus the number of secondary receivers $K$ for a minimum required base layer  SINR of $\Gamma_{\mathrm{Base}}=5$ dB,    $N_{\mathrm{T}}=8$ transmit antennas,  $J=1$ primary receiver,  and different resource  allocation schemes. It can be observed that the received interference power at each primary receiver increases with the number of secondary receivers since the secondary transmitter is required to transmit with higher power to serve extra receivers.   The proposed schemes and baseline scheme 1 are able to control the interference leakage to the primary network for any number of secondary receivers.  However,   baseline scheme 2 fails to properly control the transmit power and cannot satisfy the maximum tolerable received interference limit for  all channel realizations, due to the  non-robust resource allocation algorithm design.

\section{Conclusions}\label{sect:conclusion}
In this paper,  we studied  the robust resource allocation
algorithm design for transmit power minimization  in secure layered video transmission in secondary CR networks. The algorithm design was  formulated as a non-convex optimization problem taking into account the communication secrecy for transmission to the secondary receivers, the self-protecting structure of layered transmission,  the  imperfect knowledge of the CSI of the channels to the primary receivers,  and the interference leakage to the primary network. We showed that  the global optimal solution of the considered non-convex optimization problem can be constructed based on the primal and the dual solutions of the SDP relaxed problem. Furthermore, two suboptimal resource allocation schemes were proposed for the case when the dual problem solution is unavailable for construction of the optimal solution.  Simulation results
  unveiled the power savings enabled by the  layered transmission and the robustness of our proposed optimal scheme against the imperfect CSI of the primary receiver channels.

\section*{Appendix}

\subsection{Proof of Proposition 1}\label{sect:proof-lemma1}
Constraint C4 is non-convex due to the  log-determinant function and the coupling between optimization variables $\mathbf{W}_{l,k}$ and $\mathbf{V}$. In light of
the intractability of the constraint, we first establish a lower bound on the left hand side of  C4. Then, we will reveal the tightness of the proposed lower bound.  We now start the proof by rewriting C4  as
\begin{subequations}\label{eqn:det_0}
\begin{eqnarray} \label{eqn:det_a}
\hspace*{-6mm}\mbox{C4: }\log_2\det\Big(\mathbf{I}_{N_{\mathrm{P}_\mathrm{R}}}+\mathbf{\Sigma}_{j}^{-1}\mathbf{G}_j^H
\mathbf{W}_{1,k} \mathbf{G}_j\hspace*{-0.5mm} \Big)\hspace*{-0.5mm} \hspace*{-3mm}&\le&\hspace*{-3mm} R_{\mathrm{Eav}_{j,k}}\\
 \label{eqn:det}\hspace*{-6mm}\stackrel{(a)}{\Longleftrightarrow }\det\Big(\mathbf{I}_{N_{\mathrm{P}_\mathrm{R}}}+\mathbf{\Sigma}_{j}^{\frac{-1}{2}}\mathbf{G}_j^H
\mathbf{W}_{1,k} \mathbf{G}_j\mathbf{\Sigma}_{j}^{\frac{-1}{2}}\hspace*{-0.5mm} \Big)\hspace*{-0.5mm} \hspace*{-3mm}&\le&\hspace*{-3mm}1\hspace*{-0.5mm}+\hspace*{-0.5mm} \xi_{\mathrm{Eav}_{j,k}},
\end{eqnarray}
\end{subequations}
where  $(a)$ is due to the fact that $\mathbf{\Sigma}_{j}\succ\mathbf{0}$ and $\det(\mathbf{I}+\mathbf{AB})=\det(\mathbf{I}+\mathbf{BA})$ holds for   any  choice of matrices $\mathbf{A}$ and $\mathbf{B}$. Then, we introduce the following lemma which provides a lower bound on the left hand side of (\ref{eqn:det}).

\begin{Lem}\label{lemma:det_trace} For any square matrix $\mathbf{A}\succeq \mathbf{0}$, we have the following inequality \cite{JR:Ken_AN_PHY}:
\begin{eqnarray}
\det(\mathbf{I}+\mathbf{A})\ge 1+\Tr(\mathbf{A}),
\end{eqnarray}
\end{Lem}
where equality holds if and only if $\Rank(\mathbf{A})\le 1$.

Exploiting Lemma \ref{lemma:det_trace}, the left hand side of (\ref{eqn:det}) is lower bounded by
\begin{eqnarray}\label{eqn:det_trace}
&&\det(\mathbf{I}_{N_{\mathrm{P}_\mathrm{R}}}+\mathbf{\Sigma}_{j}^{\frac{-1}{2}}\mathbf{G}_j^H
\mathbf{W}_{1,k} \mathbf{G}_j\mathbf{\Sigma}_{j}^{\frac{-1}{2}})\notag\\
&\ge& 1+\Tr(\mathbf{\Sigma}_{j}^{\frac{-1}{2}}\mathbf{G}_j^H
\mathbf{W}_{1,k} \mathbf{G}_j\mathbf{\Sigma}_{j}^{\frac{-1}{2}}).
\end{eqnarray}
Subsequently, by combining equations (\ref{eqn:det_0}) and (\ref{eqn:det_trace}), we have the following implications
\begin{subequations}
\begin{eqnarray}
\hspace*{-10mm} &&\mbox{(\ref{eqn:det_a})}\\
\hspace*{-10mm} &&\Longleftrightarrow \mbox{(\ref{eqn:det})}\Longrightarrow \Tr(\mathbf{\Sigma}_{j}^{\frac{-1}{2}}\mathbf{G}_j^H
\mathbf{W}_{1,k} \mathbf{G}_j\mathbf{\Sigma}_{j}^{\frac{-1}{2}})\le \xi_{\mathrm{Eav}_{j,k}}\\
\hspace*{-10mm}&&\stackrel{(b)}{ \Longrightarrow } \lambda_{\max}(\mathbf{\Sigma}_{j}^{\frac{-1}{2}}\mathbf{G}_j^H
\mathbf{W}_{1,k} \mathbf{G}_j\mathbf{\Sigma}_{j}^{\frac{-1}{2}})\le \xi_{\mathrm{Eav}_{j,k}}\\
\hspace*{-10mm}&&\Longleftrightarrow\mathbf{\Sigma}_{j}^{\frac{-1}{2}}\mathbf{G}_j^H
\mathbf{W}_{1,k} \mathbf{G}_j\mathbf{\Sigma}_{j}^{\frac{-1}{2}}\preceq\xi_{\mathrm{Eav}_{j,k}}\mathbf{I}_{N_{\mathrm{P}_\mathrm{R}}}\\
\hspace*{-10mm}&&\Longleftrightarrow \mathbf{G}_j^H
\mathbf{W}_{1,k} \mathbf{G}_j\preceq\xi_{\mathrm{Eav}_{j,k}}\mathbf{\Sigma}_{j}, \label{eqn:trace_final}
\end{eqnarray}
\end{subequations}
where $(b)$ is due to $\Tr(\mathbf{A})\ge\lambda_{\max}(\mathbf{A})$ for a positive semidefinite square matrix $\mathbf{A}\succeq \mathbf{0}$. We note that  $\Tr(\mathbf{A})\ge\lambda_{\max}(\mathbf{A})$ holds if and only if  $\Rank(\mathbf{A})\le 1$. Thus, in general,  the set spanned by $\mbox{(\ref{eqn:det_a})}$ is  a subset of the set spanned by $\mbox{(\ref{eqn:trace_final})}$. Besides,  $\mbox{(\ref{eqn:det_a})}$ is equivalent to $\mbox{(\ref{eqn:trace_final})}$ when $\Rank(\mathbf{W}_{1,k})\le 1,\forall k$.

\subsection{Proof of Theorem 1}
 The proof is divided into two parts. We first study the structure of the optimal solution $\mathbf{W}^*_{l,k}$ of the relaxed version of problem (\ref{eqn:SDP_transformed}). Then, if $\exists l,k: \Rank(\mathbf{W}^*_{l,k})>1$,  we propose a method to construct a solution $\mathbf{\widetilde \Lambda}\triangleq\{\mathbf{{ \widetilde W}}_{l,k},{\mathbf{\widetilde V}},\widetilde \omega_j,\widetilde\delta_{k,j}\}$ that not only achieves the same objective value as $\mathbf{ \Lambda}^*\triangleq\{\mathbf{ W}_{l,k}^*,{\mathbf{V}}^*, \omega_j^*,\delta_{k,j}^*\}$, but also admits a rank-one beamforming matrix $\mathbf{\widetilde W}_{l,k}$.

The relaxed version of problem (\ref{eqn:SDP_transformed}) is jointly convex with respect to the optimization variables and satisfies Slater's constraint qualification. As a result, the Karush-Kuhn-Tucker (KKT) conditions are necessary and sufficient conditions \cite{book:convex} for the optimal solution of the relaxed version of problem (\ref{eqn:SDP_transformed}).   The Lagrangian function  of the relaxed version of problem (\ref{eqn:SDP_transformed}) is
\begin{eqnarray}\hspace*{-3mm}&&{\cal
L}\\ \notag
\hspace*{-3mm}&=&\hspace*{-3mm}\sum_{k=1}^K \sum_{l=1}^{L_k} \Tr(\mathbf{W}_{l,k})+\sum_{k=1}^K \sum_{l=1}^{L_k}\gamma_{l,k}\Bigg\{\Gamma_{\mathrm{req}_{l,k}}\notag\\ \notag
\hspace*{-3mm}&\times&\Big[\Tr\Big(\mathbf{H}_k\Big(\sum_{n\ne k}\sum_{l=1}^{L_n}\mathbf{W}_{l,n}+\sum_{m=l+1}^{L_k}\mathbf{W}_{m,k}-
\mathbf{W}_{l,k}\Big)\Big)\Big]\Bigg\}\\
\hspace*{-3mm}&+&\hspace*{-3mm}\sum_{t=1}^K \sum_{k\ne t}\psi_{t,k}\Bigg\{\Tr(\mathbf{H}_t\mathbf{W}_{1,k})-\Gamma_{\mathrm{tol}}\notag\\
\hspace*{-3mm}&\times&\hspace*{-3mm}\Big[\Tr\Big(\mathbf{H}_t\Big(\sum_{n\ne k}^K\sum_{l=1}^{L_n}\mathbf{W}_{l,n}+\sum_{m=2}^{L_k}\mathbf{W}_{m,k}\Big)\Big)\Big]\Bigg\}
\hspace*{-0.5mm}+\hspace*{-0.5mm}\Omega\notag\\
\hspace*{-3mm}&-&\hspace*{-3mm}  \sum_{j=1}^{J}\hspace*{-0.5mm}\Tr\hspace*{-0.5mm}\Big(\hspace*{-0.5mm} \mathbf{S}_{\mathrm{\overline{C3}}_j}\hspace*{-0.5mm}\big(\hspace*{-0.5mm}\mathbf{W}_{l,k},
\mathbf{V},\omega_j\hspace*{-0.5mm}\big)\mathbf{D}_{\mathrm{\overline{C3}}_j}\hspace*{-0.5mm}\Big)\hspace*{-0.5mm}-\hspace*{-0.5mm}\sum_{k=1}^K\hspace*{-0.5mm}\sum_{l=1}^{L_k}\hspace*{-0.5mm}\Tr\hspace*{-0.5mm}
(\mathbf{W}_{l,k}\mathbf{Y}_{l,k}\hspace*{-0.5mm})\notag\\
\hspace*{-3mm}&-&  \hspace*{-0.5mm} \sum_{k=1}^K\sum_{j=1}^{J}\hspace*{-0.5mm}\Tr\hspace*{-0.5mm}\Big(\hspace*{-0.5mm} \mathbf{S}_{\mathrm{\overline{C4}}_{k,j}}\hspace*{-0.5mm}\big(\hspace*{-0.5mm}\mathbf{W}_{1,k},
\mathbf{V},\delta_{k,j}\hspace*{-0.5mm}\big)\mathbf{D}_{\mathrm{\overline{C4}}_{k,j}}\hspace*{-0.5mm}\Big)
,\notag
\label{eqn:Lagrangian}
\end{eqnarray}
where $\Omega$ denotes the collection of the terms that only involve variables that are not relevant for the proof.  $\gamma_{l,k}\ge 0,k\in\{1,\ldots,K\},l\in\{1,\ldots,L_k\},$ and $\psi_{t,k}\ge 0,t\in\{1,\ldots,K\},
$ are the Lagrange multipliers associated with constraints C1 and C2, respectively. Matrix $\mathbf{Y}_{l,k}\succeq \mathbf{0}$ is the Lagrange multiplier matrix corresponding to the semidefinite constraint on matrix $\mathbf{W}_{l,k}$ in C6. $\mathbf{D}_{\mathrm{\overline{C3}}_j}\succeq \mathbf{0},\forall j\in\{1,\,\ldots,\,J\},$ and $\mathbf{D}_{\mathrm{\overline{C4}}_{k,j}}\succeq \mathbf{0},\forall k\in\{1,\ldots,K\},j\in\{1,\,\ldots,\,J\}$, are
the Lagrange multiplier matrices for the interference temperature constraint and the maximum tolerable SINRs of the  secondary receivers in $\mbox{\textoverline{C3}}$ and $\mbox{\textoverline{C4}}$, respectively.
In the following, we focus on the KKT conditions related to the optimal $\mathbf{W}^*_{l,k}$:
\begin{eqnarray}\label{eqn:KKT}
\mathbf{Y}^*_{l,k},\,\mathbf{D}^*_{\mathrm{\overline{C3}}_j},\,\mathbf{D}^*_{\mathrm{\overline{C4}}_{k,j}}\hspace*{-1.5mm} &\succeq&\hspace*{-1.5mm} \mathbf{0},\gamma^*_{l,k},\psi^*_{t,k}\hspace*{-0.5mm}\ge0,\\
 \mathbf{Y}^*_{l,k}\mathbf{W}^*_{l,k}\hspace*{-1.5mm} &=&\hspace*{-1mm} \mathbf{0}, \label{eqn:KKT-complementarity}\\
\nabla_{\mathbf{W}^*_{l,k}}{\cal
L}\hspace*{-1.5mm} &=&\hspace*{-1.0mm} \mathbf{0}, \label{eqn:KKT-gradient}
\end{eqnarray}
where $\mathbf{Y}^*_{l,k},\mathbf{D}^*_{\mathrm{\overline{C3}}_j},\mathbf{D}^*_{\mathrm{\overline{C4}}_{k,j}},\gamma^*_{l,k},$ and $\psi^*_{t,k}$, are the optimal Lagrange multipliers for the dual problem of (\ref{eqn:SDP_transformed}). From the complementary slackness condition in (\ref{eqn:KKT-complementarity}),  we observe that the columns of $\mathbf{W}^*_{l,k}$ are required to lie in the null space of $\mathbf{Y}^*_{l,k}$ for $\mathbf{W}^*_{l,k}\ne \mathbf{0}$.  Thus, we study  the  composition of $\mathbf{Y}^*_{l,k}$ to obtain the structure of  $\mathbf{W}^*_{l,k}$. The KKT condition in (\ref{eqn:KKT-gradient}) leads to
\begin{eqnarray}\label{eqn:KKT-gradient-equivalent}
\mathbf{Y}^*_{l,k}\hspace*{-2mm}&=&\hspace*{-2mm}
\mathbf{A}_{l,k}-\Big[\gamma_{l,k}^*-\sum_{t<l}\gamma_{t,k}^*\Gamma_{\mathrm{req}_{t,k}}\Big]\mathbf{H}_k\\
\mathbf{A}_{l,k}\hspace*{-2mm}&=&\hspace*{-2mm}
 \left\{\begin{array}{ll}\mathbf{B}_{l,k}\hspace*{-0.5mm}+\hspace*{-0.5mm} \mathbf{C}_{t,k}
 &\mbox{if $l=1$} \\
 \mathbf{B}_{l,k}- \underset{t\ne k}\sum \psi_{t,k}^*\Gamma_{\mathrm{tol}}\mathbf{H}_t&\mbox{otherwise}
       \end{array} \right.\\
        \mathbf{C}_{t,k}\hspace*{-2mm}&=&\hspace*{-2mm}\underset{t\ne k}\sum \psi_{t,k}^*\mathbf{H}_t\hspace*{-0.5mm}+\hspace*{-0.5mm} \overset{K}{\underset{k=1}\sum}\overset{J}{\underset{j=1}\sum}\mathbf{R}_{j}\mathbf{D}_{\mathrm{\overline{C4}}_{k,j}}^*\mathbf{R}_{j}^H,\\
   \mathbf{B}_{l,k}\hspace*{-2mm}&=&\hspace*{-2mm} {\underset{m\ne k}\sum}  \overset{L_m} {\underset{r=1}\sum} \gamma_{r,m}^*\Gamma_{\mathrm{req}_{r,m}}-\Gamma_{\mathrm{tol}}\Big[ {\underset{t\ne k}\sum}  {\underset{n\ne t,k}\sum}\psi_{t,n}^*\mathbf{H}_t \Big]\notag\\
 &\hspace*{-0.5mm}+\hspace*{-0.5mm}& \overset{J} {\underset{j=1}\sum}\overset{N_{\mathrm{R}}} {\underset{q=1}\sum}\Big[    \mathbf{U}_{\mathbf{g}_j}\mathbf{D}_{\mathrm{\overline{C3}}_{j}}^*\mathbf{U}_{\mathbf{g}_j}^H\Big]_{a:b,c:d}.
\end{eqnarray}
Subscripts $a,b,c,d$ are given by $a=(q-1)N_{\mathrm{T}}+1,b=q N_{\mathrm{T}},c=(q-1)N_{\mathrm{T}}+1,$ and $d=q N_{\mathrm{T}}$, respectively.
 Without loss of generality, we define $r_{l,k}=\Rank(\mathbf{A}^*_{l,k})$ and the orthonormal basis of the null space of $\mathbf{A}^*_{l,k}$ as $\mathbf{\mathbf{\Upsilon}}\in\mathbb{C}^{N_{\mathrm{T}}\times (N_{\mathrm{T}}-r_{l,k})}$ such that $\mathbf{A}^*_{l,k}\mathbf{\Upsilon}_{l,k}=\mathbf{0}$ and $\Rank(\mathbf{\Upsilon}_{l,k})=N_{\mathrm{T}}-r_{l,k}$. Let ${\boldsymbol \phi}_{\tau_{l,k}}\in \mathbb{C}^{N_{\mathrm{T}}\times 1}$, $1\le \tau_{l,k}\le N_{\mathrm{T}}-r_{l,k}$, denote the $\tau_{l,k}$-th column  of $\mathbf{\Upsilon}_{l,k}$. By exploiting \cite[Proposition 4.1]{JR:rui_zhang},  it can be shown that $\Big[\gamma_{l,k}^*-\sum_{t<l}\gamma_{t,k}^*\Gamma_{\mathrm{req}_{t,k}}\Big]\mathbf{H}_k\ne \zero$ and $\mathbf{H}_k\mathbf{\Upsilon}_{l,k}=\mathbf{0}$ for the optimal solution. Besides, we can express the optimal solution of  $\mathbf{W}^*_{l,k}$  as
\begin{eqnarray}\label{eqn:general_structure}
\mathbf{W}^*_{l,k}=\sum_{\tau_{l,k}=1}^{N_{\mathrm{T}}-r_{l,k}} \alpha_{\tau_{l,k}} {\boldsymbol \phi}_{\tau_{l,k}}  {\boldsymbol \phi}_{\tau_{l,k}} ^H  + \underbrace{f_{l,k}\mathbf{u}_{l,k}\mathbf{u}^H_{l,k}}_{\mbox{rank-one}},
\end{eqnarray}
where $\alpha_{\tau_{l,k}}\ge0, \forall \tau_{l,k}\in\{1,\ldots,N_{\mathrm{T}}-r_{l,k}\},$ and $f_{l,k}>0$ are positive scalars and $\mathbf{u}_{l,k}\in \mathbb{C}^{N_{\mathrm{T}}\times 1}$, $\norm{\mathbf{u}_{l,k}}=1$, satisfies $\mathbf{u}^H_{l,k}\mathbf{\Upsilon}_{l,k}=\mathbf{0}$. In particular, we have the following equality:
\begin{eqnarray}\label{eqn:general_structure}
\mathbf{H}_k\mathbf{W}^*_{l,k}=\underbrace{\sum_{\tau_{l,k}=1}^{N_{\mathrm{T}}-r_{l,k}}\mathbf{H}_k \alpha_{\tau_{l,k}} {\boldsymbol \phi}_{\tau_{l,k}}  {\boldsymbol \phi}_{\tau_{l,k}} ^H}_{=\zero}  + \mathbf{H}_k f_{l,k}\mathbf{u}_{l,k}\mathbf{u}^H_{l,k}.
\end{eqnarray}

In the second part of the proof, we construct another solution $\mathbf{\widetilde \Lambda}\triangleq\{\mathbf{{ \widetilde W}}_{l,k},{\mathbf{\widetilde V}},\widetilde \omega_j,\widetilde\delta_{k,j}\}$ based on  (\ref{eqn:general_structure}).  Suppose there exist pair of $l$ and $k$ such that $\Rank(\mathbf{W}^*_{l,k})>1$. Let
\begin{eqnarray}\label{eqn:rank-one-structure}\mathbf{\widetilde W}_{l,k}\hspace*{-2mm}&=&\hspace*{-2mm}f_{l,k}\mathbf{u}_{l,k}
\mathbf{u}^H_{l,k}=\mathbf{W}^*_{l,k}\hspace*{-0.5mm}-\hspace*{-0.5mm}\sum_{\tau_{l,k}=1}^{N_{\mathrm{T}}-r_{l,k}} \alpha_{l,k} {\boldsymbol \phi}_{\tau_{l,k}} {\boldsymbol \phi}_{\tau_{l,k}}^H,\\
 \mathbf{\widetilde V}\hspace*{-2mm}&=&\hspace*{-2mm}\mathbf{ V}^*\hspace*{-0.5mm}+\hspace*{-1.5mm}\sum_{\tau_{l,k}=1}^{N_{\mathrm{T}}-r_{l,k}} \alpha_{l,k} {\boldsymbol \phi}_{\tau_{l,k}} {\boldsymbol \phi}_{\tau_{l,k}}^H,\quad
\widetilde\omega_j=\delta^*_j,\,
\widetilde\delta_{k,j}=\delta_{k,j}^*.\,\notag
\end{eqnarray}
Then, we substitute the constructed solution $\mathbf{\widetilde \Lambda}$ into the objective function and the constraints in (\ref{eqn:SDP_transformed}) which yields \eqref{eqn:equivalent_objective} on the top of next page.

\begin{figure*}[ht]
\begin{eqnarray}\label{eqn:equivalent_objective}
 &&\hspace*{-10mm}\mbox{Objective value:}\quad \sum_{k=1}^K\sum_{l=1}^{L_k}\Tr(\mathbf{\widetilde W}_{l,k})+\Tr(\mathbf{\widetilde V})
 =\sum_{k=1}^K\sum_{l=1}^{L_k}\Tr(\mathbf{W}_{l,k}^*)+\Tr(\mathbf{V}^*) \\
 \label{eqn:C1_new_solution}
 \mbox{C1:}&&
 \hspace*{-5mm} \frac{\Tr\Big(\mathbf{H}_k\big(\mathbf{W}^*_{l,k}-\sum_{\tau_{l,k}=1}^{N_{\mathrm{T}}-r_{l,k}} \alpha_{l,k} {\boldsymbol \phi}_{\tau_{l,k}} {\boldsymbol \phi}_{\tau_{l,k}}^H\big)\Big)}{\Tr\Big(\mathbf{H}_k \big(\overset{K}{\underset{n\ne k}{\sum}}\overset{L_n}{\underset{r=1}{\sum}}\mathbf{W}_{r,n}^*+\overset{L_k}{\underset{m=l+1}{\sum}} \mathbf{W}_{m,k}^*\big)\Big)+\Tr\Big(\mathbf{H}_k(\mathbf{ V}^*+\sum_{\tau_{l,k}=1}^{N_{\mathrm{T}}-r_{l,k}} \alpha_{l,k} {\boldsymbol \phi}_{\tau_{l,k}} {\boldsymbol \phi}_{\tau_{l,k}}^H)\Big)+ \sigma_{\mathrm{s}_k}^2}\notag\\
  =&&\hspace*{-5mm}  \frac{\Tr(\mathbf{H}_k\mathbf{ W}_{l,k}^*)}{\Tr\Big(\mathbf{H}_k \big(\overset{K}{\underset{n\ne k}{\sum}}\overset{L_n}{\underset{r=1}{\sum}}\mathbf{W}_{r,n}^*+\overset{L_k}{\underset{m=l+1}{\sum}} \mathbf{W}_{m,k}^*\big)\Big)+\Tr(\mathbf{H}_k\mathbf{ V}^*)+\sigma_{\mathrm{s}_k}^2}\ge  \Gamma_{\mathrm{req}_k},\forall l, k,\notag\\
  \mbox{C2:}\notag&&\hspace*{-5mm}\frac{\Tr(\mathbf{H}_t\mathbf{\widetilde W}_{1,k})}{   \Tr\Big(\mathbf{H}_t\big(\overset{K}{\underset{n\ne t}{\underset{n\ne k}{\sum}}}\overset{L_n}{\underset{r=1}{\sum}}\mathbf{W}_{r,n}^*+\overset{L_k}{\underset{m=2}{\sum}}\mathbf{W}_{m,k}^*\big)\Big) +\Tr(\mathbf{H}_t\mathbf{\widetilde V })+\sigma_{\mathrm{s}_k}^2}\notag\\
    \le&&\hspace*{-5mm}   \frac{\Tr(\mathbf{H}_t\mathbf{ W}_{1,k}^*)}{   \Tr\Big(\mathbf{H}_t\big(\overset{K}{\underset{n\ne t}{\underset{n\ne k}{\sum}}}\overset{L_n}{\underset{r=1}{\sum}}\mathbf{W}_{r,n}^*+\overset{L_k}{\underset{m=2}{\sum}}\mathbf{W}_{m,k}^*\big)\Big) +\Tr(\mathbf{H}_t\mathbf{ V^*} )+\sigma_{\mathrm{s}_k}^2}
   \le \Gamma_{\mathrm{tol}},\forall t\ne k, t\in\{1,\ldots,K\},\notag\\
 \mbox{\textoverline{C3}:}\notag&&\hspace*{-5mm} \mathbf{S}_{\mathrm{\overline{C3}}_j}(\mathbf{\widetilde W}_{l,k},\mathbf{\widetilde V},\widetilde\omega_j)\succeq\mathbf{S}_{\mathrm{\overline{C3}}_j}(\mathbf{W}^*_{l,k},{\mathbf{V}}^*, \omega_j^*)\notag\\
 +&&\hspace*{-5mm} \mathbf{U}_{\mathbf{g}_j}^H\Big[\sum_{\tau_{l,k}=1}^{N_{\mathrm{T}}-r_{l,k}} \mathbf{I}_{N_{\mathrm{P}_\mathrm{R}}}\otimes \alpha_{l,k} {\boldsymbol \phi}_{\tau_{l,k}} {\boldsymbol \phi}_{\tau_{l,k}}^H\Big] \mathbf{U}_{\mathbf{g}_j} \succeq \mathbf{0},\forall j\in\{1,\ldots,J\},\notag\\
\mbox{\textoverline{C4}:}&&\hspace*{-5mm}\mathbf{S}_{\mathrm{\overline{C4}}_{k,j}}(\mathbf{\widetilde W}_{l,k},\mathbf{\widetilde V},\widetilde\omega_j)\succeq\mathbf{S}_{\mathrm{\overline{C4}}_{k,j}}(\mathbf{W}^*_{l,k},{\mathbf{V}}^*, \delta_{k,j}^*)+\mathbf{R}_{j}^H\Big[\sum_{\tau_{l,k}=1}^{N_{\mathrm{T}}-r_{l,k}} \alpha_{l,k} {\boldsymbol \phi}_{\tau_{l,k}} {\boldsymbol \phi}_{\tau_{l,k}}^H\Big]\mathbf{R}_{j} \succeq \mathbf{0},\forall k,j,\notag \\
\mbox{C5:}&&\hspace*{-5mm} \mathbf{\widetilde{W}}_{l,k}\succeq \mathbf{0},\quad\mbox{{C6}:}\,\mathbf{\widetilde{V}}\succeq \mathbf{0},
\quad\mbox{{C8}: }\widetilde\omega_j=\omega_j^*\ge0, \quad\mbox{{C9}: }\widetilde\delta_{i,k}=\delta_{j,k}^*\ge0.\label{eqn:C10_new_solution}\notag
\end{eqnarray}\hrulefill
\end{figure*}
It can be seen from (\ref{eqn:equivalent_objective}) that
the constructed solution set achieves the same optimal value
as the optimal solution  while satisfying all the constraints.
Thus, $\mathbf{\widetilde \Lambda}$ is also an optimal solution of (\ref{eqn:SDP_transformed}). Besides, the
constructed beamforming matrix $\mathbf{\widetilde W}_{l,k}$ is a rank-one matrix, i.e.,
$\Rank(\mathbf{\widetilde W}_{l,k}) = 1$. On the other hand, we can obtain the values
of $f_{l,k}$ and $\alpha_{l,k}$ in (\ref{eqn:rank-one-structure}) by substituting the variables in (\ref{eqn:rank-one-structure}) into
the relaxed version of (\ref{eqn:SDP_transformed}) and solving the resulting convex
optimization problem for $f_{l,k}$ and $\alpha_{l,k}$.

If there is more than one pair of $l$ and $k$ such that  $\Rank(\mathbf{W}_{l,k})>1$, then we employ (\ref{eqn:rank-one-structure}) more than once and construct the rank-one solution. Besides, the ordering of the $l$ and $k$ pairs  in constructing the optimal solution does not affect to the optimal objective value.


\end{document}